\documentclass[11pt, article, cspaper, onecolumn, compsoc, reqno]{IEEEtran}
\usepackage{amssymb,amsfonts,amsmath,amsthm}
\usepackage{graphicx}
\usepackage{caption}
\usepackage{graphicx, wrapfig, lipsum}
\usepackage{enumerate}
\usepackage{setspace}
\usepackage{array}
\usepackage{epsfig}
\usepackage{float}
\usepackage{subfigure}
\usepackage[margin=1in]{geometry}
\usepackage{amsmath}
\usepackage{hyperref}
\usepackage{epstopdf}
\hypersetup{
    colorlinks,%
    citecolor=red,%
    filecolor=red,%
    linkcolor=red,%
    urlcolor=blue
   }
\numberwithin{equation}{section} \numberwithin{theorem}{section}

\def\bt{\begin{thm}}
\def\et{\end{thm}}

\def\bl{\begin{lem}}
\def\el{\end{lem}}

\def\bd{\begin{defi}}
\def\ed{\end{defi}}

\def\bc{\begin{cor}}
\def\ec{\end{cor}}

\def\bp{\begin{proof}}
\def\ep{\end{proof}}

\def\br{\begin{rem}}
\def\er{\end{rem}}

\newtheorem{thm}{Theorem}[section]
\newtheorem{lem}{Lemma}[section]
\newtheorem{defi}{Definition}[section]
\newtheorem{rem}{Remark}[section]
\newtheorem{cor}{Corollary}[section]

\newcommand{\be}{\begin{equation}}
\newcommand{\ee}{\end{equation}}
\newcommand{\ba}{\begin{aligned}}
\newcommand{\ea}{\end{aligned}}

\newlength{\figurewidth}
\newlength{\smallfigurewidth}

\ifCLASSOPTIONcompsoc
\else
\fi

\ifCLASSINFOpdf

\else

\fi
\begin{document}
\date{}
\title{Wavelet-Based Compression and Peak Detection Method for the Experimentally Estimation of Microtubules Dynamic Instability Parameters Identified in Three States}

\author{Shantia~Yarahmadian,~Vineetha~Menon,~Majid~Mahrooghy,~and Vahid~A. Rezania
\IEEEcompsocitemizethanks{
\protect
\IEEEcompsocthanksitem Shantia~Yarahmadian is with the Department of Mathematics, Mississippi State University. E-mail: syarahmadian@math.msstate.edu
\IEEEcompsocthanksitem Vineetha~Menon is with the Department of Electrical and Computer Engineering, Mississippi State University. E-mail: vk132@msstate.edu
\IEEEcompsocthanksitem Vahid~Rezania is with the Department of Physical Sciences, Macewan University, Edmonton, Canada. E-mail: rezaniav@macewan.ca
\IEEEcompsocthanksitem Majid~Mahrooghy is with Xoran Technologies. E-mail: majid.mahrooghy@gmail.com
}
\thanks{}}

\maketitle

\begin{abstract}
Recent studies has revealed that Microtubules (MTs) exhibit three transition states of growth, shrinkage and pause. In this paper, we first introduce a three states random evolution model as a framework for studying MTs dynamics in three transition states of growth, pause and shrinkage. Then, we introduce a non-traditional stack run encoding scheme with 5 symbols for detecting transition states as well as to encode MT experimental data. The peak detection is carried out in the wavelet domain to effectively detect these three transition states. One of the added advantages of including peak information while encoding being that it enables to detect the peaks efficiently and encodes them simultaneously in the wavelet domain without having the need to do further processing after the decoding stage. Experimental results show that using this form of non-traditional stack run encoding has better compression and reconstruction performance as opposed to traditional stack run encoding and  run length encoding schemes. Parameters for MTs modeled in the three states are estimated and is shown to closely approximate original MT data for lower compression rates. As the compression rate increases, we may end up throwing away details that are required to detect transition states of MTs. Thus, choosing the right compression rate is a trade-off between admissible level of error in signal reconstruction, its parameter estimation and considerable rate of compression of MT data.

\end{abstract}


\label{sec:intro}
Microtubules (MTs) are natural polymers in the inside of living eukaryotic cells assembled by aggregation of alpha-beta tubulin protein dimers. Under proper conditions, when the tubulin concentration is above a threshold level, tubulin dimers bind to each other and construct a tube-like structure with a typical of $13$ subunits in a cross-section \cite{1}. 
MTs play a major role in providing the structural stability of the cell as well as serving as a semi-rigid structural element in the cell for intracellular distribution networks. MTs participate in chromosome segregation before cell division, neuronal activities, and cell motility \cite{2}. They also play a key role in various diseases such as Alzheimer disease \cite{3} and Parkinson disease \cite{4} and different form of cancers \cite{5}.
\\\\
Formation of MTs within cells is typically initiated from a complex of tubulin proteins as a nucleation template. It was first observed that in time, whether in the living cell ({\it in vivo}) or in the lab with the purified tubulin ({\it in vitro}), MTs display stochastic fluctuations between polymerization and depolymerization states. This process is termed as \textit{dynamic instability} \cite{6}.  More recently, Ambrose and Wasteneys \cite{7} and Shaw et al. \cite{8} observed that MTs are likely to spend some times in a third state, an attenuated or pause state, where they exhibit little or no detectable length changes over time. including this third state of pause, MTs dynamic instability is characterized by the transition frequencies between three states of growth ($g$), pause ($p$) and shrinkage ($s$). The transition frequencies are modulated during certain cellular processes to maintain the length distribution and density of the polymer array.
\\\\
Mechanisms governing dynamic instability are still an active subject of both experimental and theoretical investigations.
The stochastic nature of the transition between polymerization-pause-depolymerization states complicates the construction of deterministic equations that can be numerically solved for properties such as MT length distribution. A closed form kinetic description of microtubule dynamics has not been successfully rendered to date. However, several analytical models have been developed to describe how the major factors leading to dynamic instability (e.g. growth and shortening velocities together with catastrophe and rescue frequencies) will produce a steady state system of polymers under various conditions. See for example \cite{9, 10, 11, 12, 13, 14, 15} and the references cited in them.
\\\\
In order to elucidate the MTs' dynamic instability, they have been subject of extensive experimental studies 
using optical and electronic microscopy techniques. Time-lapse video microscopy analysis demonstrated the 
existence of many stages in the development of a microtubule, including an initial nucleation stage, followed by a more 
or less continuous growth mode that is then stochastically interrupted by sudden and catastrophic disassembly 
that can again be followed by a growth stage. To understand and capture dynamic behavior of MTs, a high temporal resolution observation is necessary.
\\\\
Although with new technological developments one can achieve beyond microsecond frame rates in video microscopy, 
the major limitation on collecting data is still in damaging the specimen or probe over the time course of the experiment.
It is clear that the specimen can only be illuminated for a finite period of time, and as a result, the collected time-lapse data in comparison to the time scale of the intracellular events are often being sparse. Furthermore, there are some substantial practical limitations on the sampling of the MT parameters such as experimental equipment precision and spatiotemporal resolution. Therefore, a use of data analyzing technique is necessary to improve the resolution and the Signal-to-Noise Ratio (SNR) of the time-lapse observations.
\\\\
Several compression methods for various biomedical signals have been proposed in the literature \cite{16} , such as Amplitude Zone-Time Epoch Coding (AZTEC) \cite{17}, Turning point (TP) \cite{16}, Coordinate Reduction Time Encoding System (CORTES) \cite{18}, Wavelet based compression methods such as Set Partitioning in Hierarchical Trees (SPIHT) \cite{19}, Embedded Zerotree Wavelets (EZW) \cite{20}, JPEG2000 \cite{21} are few of the image compression algorithms adopted from the realm of image processing to biomedical signals. All of the wavelet based compression algorithms appear lucrative because of the sparse coefficients in the wavelet domain that  provide desirable qualities such as sparsity, inherent noise and redundancy reduction. 
\\\\
Recently, Mahrooghy et al. \cite{12} discussed the use of compressing sensing (CS) on the sampled data of the microtubule length time series.  The aim was to reconstruct the original signal with relatively low sampling rates. They modeled the data as a Dichotomous Markov Noise (DMN), a two states stochastic model, that considers an MT either grows or shortens.   They implemented the method on experimental data sets where  their results showed that combining CS along with the wavelet transform significantly reduces the recovery errors comparing in the absence of wavelet transform, especially in the low and the medium sampling rates.
\\\\
In this paper we are studying the dynamic instability of microtubules in the three states of stochastic polymerization, pause and depolymerization. For this purpose we have used random evolution velocity model in three states, which is also known as Trichotomous Markov Noise (TMN) \cite{22}. The geometry of the study is a semi-infinite geometry, in which infinitely rigid microtubules grow perpendicularly to a nucleating planar surface along the $x$-axis. In our model infinitely rigid microtubules are growing perpendicular to a nucleation planar surface and randomly switching between three states of growth (polymerization) state ($g$), pause ($p$) and shrinkage (depolymerization) state ($s$). In this framework, we have explored a non traditional stack run coding \cite{23, 24} to which encodes additional peak information during the encoding  process reducing the overhead in the decoding process and helps to reconstruct both the signal and extract the transition features of Microtubules. The reconstructed MT signal is used to estimate the TMN parameters and resultant error in the reconstruction process.

\begin{figure}[h!]
\centering
\includegraphics[scale=.5]{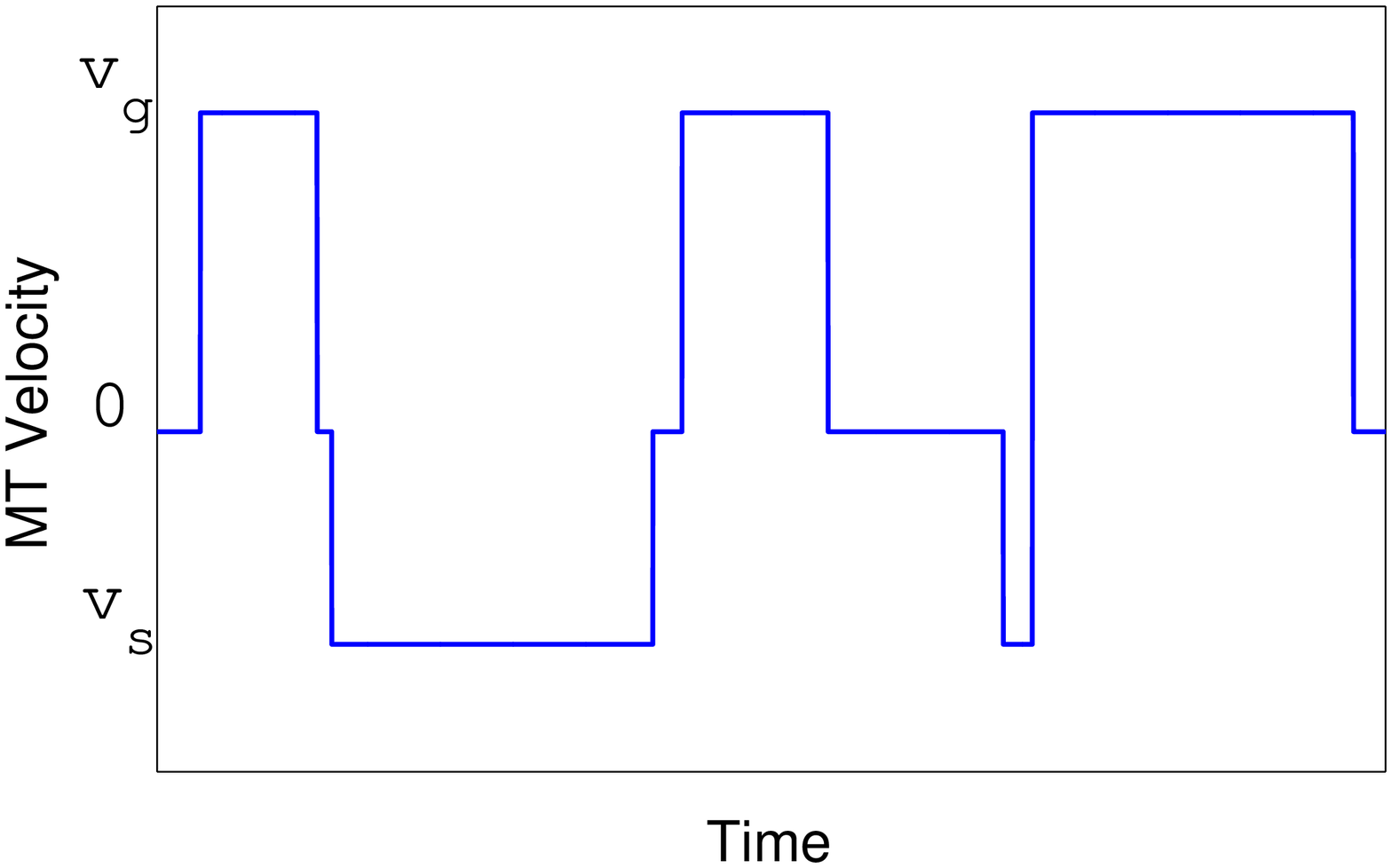}
\label{fig:TMN}
\end{figure}

\section{Trichotomous Markov Noise Model for Microtubules length Time Series}
\label{sec:TMN}
The three-state dynamic instability model involves eight parameters: six transition rates between states of growth, pause and shrinkage represented by $f_i^j$ where $i, j \in \{g, p, s\}$ and growth and shrinkage rates $v_g$ and $v_s$. As a simplification of this three-state model, a two-state dynamical instability model involving four parameters has been studied \cite{} and used extensively in cell biology \cite{13, 14, 15, 25}. The mean length and mean lifetime of a MT in the two-state model depend on a threshold quantity $f_{gs}$, $f_{sg}$, $v_g$, and $v_s$ \cite{25}. If the quantity is positive, the MTs tend to shrink more than they grow, and the MTs will have a finite mean length and mean lifetime. Otherwise, on average, they tend to grow forever. For the three-state case, we compute the equivalent equations as it is described below.
In this case the stochastic velocity is a Trichotomous Markov Noise (TMN) that takes values $v_g$
with probability ${p}_g$, $0$ with probability ${p}_s$
and $v_s$ with probability ${p}_s$. In this case:
\be 
\label{Tri}
\frac{d}{dt}\begin{pmatrix} {p}_g\\{p}_p\\{p}_s\end{pmatrix}=
\begin{pmatrix} -f_{gs}-f_{gp}&f_{pg}&f_{sg}
\\f_{gp}&-f_{pg}-f_{ps}&f_{sp}\\f_{gs} & f_{ps}&-f_{sg}-f_{sp}\end{pmatrix} \begin{pmatrix} {p}_g\\{p}_p\\{p}_s\end{pmatrix}
\ee
\subsection{Notations:}
Throughout this paper we use following notations for simplicity. 
\be \label{F1}
\ba
&F_g=f_{gp}+f_{gs}\\
&F_p=f_{pg}+f_{ps}\\
&F_s=f_{sp}+f_{sg}\\
&F=F_s+F_p+F_g\\
F_{gp}&=f_{gs}f_{pg}+f_{gs}f_{ps}+f_{gp}f_{ps}\\
F_{sp}&=f_{sg}f_{pg}+f_{sg}f_{ps}+f_{sp}f_{pg}\\
F_{sg}&=f_{sp}f_{gs}+f_{sp}f_{gp}+f_{sg}f_{gp}\\
\Omega&=F_{sp}+F_{pg}+F_{sg}
\ea
\ee
\begin{rem}
$f_{ij}, i, j \in \{g, p, s\}$ represent the corresponding transition frequencies per unit time.
\end{rem}

\subsection {Waiting times and Steady States}
\label{sec:wait_time}
\begin{lem}
In steady-state the average velocity reaches:
\be
V=\frac{v_gF_{sp}-v_sF_{pg}}{\Omega}
\ee
\end{lem}
\begin{proof}
The frequency matrix
$$
A=\begin{pmatrix} -f_{gs}-f_{gp}&f_{pg}&f_{sg}
\\f_{gp}&-f_{pg}-f_{ps}&f_{sp}\\f_{gs} & f_{ps}&-f_{sg}-f_{sp}\end{pmatrix}
$$
has the following eigenvalues:
\be
\lambda_{1, 2, 3}=0, \frac{-F \pm \sqrt{F^2-4\Omega}}{2}
\ee
\\
Therefore, the solution of \eqref{Tri} is written as:
\be
{p}(t) = \sum_{i=1}^3C_iV_ie^{\lambda_i t}
\ee

where, $\lambda_i$'s and $V_i$'s are the corresponding eigenvalue and eigenvectors. $C_i$'s are constants depending on the transition frequencies and the initial probabilities of growth, pause and shrinkage. By the substitution of the corresponding eigenvectors, and using the fact that ${p}_g(t)+{p}_p(t)+{p}_s(t)=$1, the stationary probabilities are written as:
\medskip
\be
{p}_g(\infty)=\frac{F_{sp}}{F_{sp}+F_{sg}+F_{pg}}
\ee
\medskip
\be
{p}_p(\infty)=\frac{F_{sg}}{F_{sp}+F_{sg}+F_{pg}}
\ee
\medskip
\be
{p}_s(\infty)=\frac{F_{pg}}{F_{sp}+F_{sg}+F_{pg}}
\ee
\\
Now, the instantaneous expected value for the velocity can be calculated as:
\be
\bar{v}(t)=v_g{p_g}(t)-v_s{p_s(t)}
\ee
In the steady-state, this average will reach:
\be
V=\lim_{t \to \infty} \bar{v}(t)=\frac{v_gF_{sp}-v_sF_{pg}}{\Omega}
\ee
\end{proof}
\begin{rem}
In the unlimited regime, $V>0$, the MT continuously grows in time with the average speed $V>0$ and the growth is the dominating event, while in the limited regime, $V<0$, the shrinkage is the dominating event. The transition between the two regimes happen when $V=0$, where there is a balance between growth and shrinkage rates that gives a steady-state length of MT. This means that the growth regime of MT depends on local value of $V$. The MT reaches the steady state when:
\be
v_gF_{sp}=v_sF_{pg}
\ee
\end{rem}

\subsection{Master Equation and Average Length}
\label{sec:Avg_len}
Following \cite{14,15}, the probability densities, $P_g(x, t)$, $P_s(x, t)$, and $P_p(x, t)$ are defined as the probability of finding at time $t$ the MT in the phases, growth, shrinkage and pause, respectively with its free length at $x$. The dynamic instability is written as:

\be 
\label{tmas}
\begin{cases}
\ba
\frac{\partial P_g}{\partial t}&=-v_g\frac{\partial P_g}{\partial x}-(f_{gs}+f_{gp})P_g+f_{pg}P_p+ f_{sg}P_s
\\\\
\frac{\partial P_p}{\partial t}&=f_{gp}P_g-(f_{pg}+f_{ps})P_p+ f_{sp}P_s
\\\\
\frac{\partial P_s}{\partial t}&=v_s\frac{\partial P_s}{\partial x}+f_{gs}P_g+f_{ps}P_p- (f_{sg}+f_{sp})P_s
\ea
\end{cases}
\ee
\\\\
\begin{lem}
The average length distribution of MT is an exponential distribution with the mean of:
\be
<L>=\frac{v_sv_g}{v_s\frac{F_{gp}}{F_p}-v_g\frac{F_{ps}}{F_p}}
\ee
\end{lem}
\begin{proof}
In steady-state, by setting the time partial derivative equal to zero and substitution of $P_p(t)$, \eqref{tmas} will convert to:
\be
\label {tmas1}
\begin{cases}
\ba
-v_g\frac{\partial P_g}{\partial x}-\Big(\frac{F_{gp}}{F_p}\Big)P_g+\Big(\frac{F_{sp}}{F_p}\Big)P_s=0\\
v_s\frac{\partial P_s}{\partial x}+\Big(\frac{F_{gp}}{F_p}\Big)P_g-\Big(\frac{F_{sp}}{F_p}\Big)P_s=0
\ea
\end{cases}
\ee
Simplifying \eqref{tmas1} will result in:
\be
\frac{\partial^2 P_i}{\partial x^2}(x)+\Bigg[\frac{v_s\Big(\frac{F_{gp}}{F_p}\Big)-v_g\Big(\frac{F_{ps}}{F_p}\Big)} {v_sv_g}\Bigg]\frac{\partial P_g}{\partial x}(x)=0,~~~i\in \{g, s\}
\ee 

Now, using the fact that $\lim_{x \to \infty} P_{i}(x)=0$ and $\int_{0}^{\infty}P_{i}(x)dx=1$, $i\in \{g, s\}$, we get the exponential distribution with the average length as mentioned.
\end{proof}

\section{Wavelet Transform}   
\label{sec:Wavelet}
The wavelet transform is a useful computational tool for processing signal and images. It has been used for variety of applications such as compression, filtering, denoising, and feature extraction. Unlike the Fourier transform in which the basis function are sinusoid functions, wavelet transform is based on small wave with varying frequency and limited time duration. This provides the freedom to choose the time-limited basis functions based on desirable properties such as symmetry, phase shift, orthogonal, biorthogonal so on. Equation (3.1) describes a 1D-discrete time function, which is expanded, by wavelet basis functions in the time and frequency domain as \cite{28}:
\be
\label{6}
f(n)=\sum_k\sum_j a_{j, k}\psi_{j, k}(n)
\ee
where $f(n) \in L_2 (R)$, $k, j\in Z$. $k$ stands for time and $j$ represents the frequency (or translating variables and scale, respectively) \cite{28}. The basis functions are defined as:
\be
\label{6.1}
\psi_{j,k} (n)= 2^{\frac{j}{2}}\psi(2^j n-k) 
\ee
where $\psi(t)$ is the mother wavelet and the $a_{j,k}$ are the corresponding wavelet coefficients. Equation \eqref{6} can also be expanded as the following formula: 
\be
f(n)=\sum_k c_{j_0}(k)\phi_{j_0,k}(n)+\sum_{j=j_0}\sum_k d_j(k)\psi_{j, k}(n)
\ee
The first term of $f(n)$ corresponds to the coarse resolution while the second term represents  the detail (or wavelet) resolution of the signal. $c_{j_0}(k)$  and  $d_j (k)$  are the corresponding approximation (or scale) and detail (or wavelet) coefficients at scale $j$, respectively. These coefficients can be calculated by using the following equations:
\be
\label{7}
c_{j_0}(k)= \langle f(n),\phi_{j_{0},k}(n) \rangle  
\ee
\be
d_{j}(k)= \langle f(n),\psi_{j,k}(n) \rangle  
\ee
Note that $j_0$ is an arbitrary starting scale. Also in this paper the maximum scale $j$ is considered as the decomposition level. For this work, we have used Daubechies wavelet $db4$ as our mother wavelet with $8$ levels of wavelet decomposition \cite{28}.

\section{Stack Run Encoding}   
\label{sec:SRC}
Stack Run Coding is a wavelet based compression technique that aims on exploiting  sparsity in wavelet domain to achieve excellent compression \cite{29}. It can be viewed as performing a varied form  of run length coding technique on wavelet domain. Lets consider the wavelet transformed MT signal to which an uniform quantizer with a step size $q$  is applied across all the levels or subbands. The quantization process discretizes the wavelet coefficients enabling us to partition them into significant (non-zero valued) and non-significant (zero valued) coefficients. We then raster scan the coefficients from lowest level to highest wavelet level encoding the binary representation of the coefficients to its equivalent code using 5 symbols ($+, -, 0, 1, p$). The non-zero coefficients are arranged from their most significant bit (MSB) to their least significant bit (LSB). Stack Run Coding utilizes the intuitiveness that MSB of a code need not be explicitly coded, since its always $1$. Hence, to determine the remaining codeword, all we have to do is append $1$ ahead of the remaining coded bits. During raster scan of level $1$ (lowest level)  of wavelet coefficients, we use energy packing efficiency (EPE) \cite{30} to find the location of inherent peaks in wavelet domain. We incorporate the knowledge of peaks during encoding, so that we can directly use to the peak information to extract MT parameters while decoding.\\

We encode the quantized wavelet coefficients using the 5 symbols ($+, -, 0, 1, p$) as follows:

\begin{enumerate}
\item \emph{Symbol $+$}: Symbol $+$ is not only used to represent the sign of MSB of a positive significant wavelet coefficent, but  also to replace and implicitly represent $1$  present in the MSB.\\
\item \emph{Symbol $-$}: Symbol $-$ is used to represent the sign of MSB of a negative significant wavelet coefficient, as well as replace and implicitly represent zero value in run lengths. That is, if a run length has a value $0$. Then it is correspondingly coded as $-$. If a run length has a value $2$. Then its corresponding binary code would be $1 0$ and it is encoded as $ + - $. Where $1$ is the MSB coded as $+$ and $0$ is coded as $- $.\\
\item \emph{Symbol $0$}:  Symbol $0$ is used to represent the binary value $0$ occurring in the binary representation of the wavelet coefficients.\\
\item \emph{Symbol $1$}: Symbol $1$ is used to represent the binary value $1$ (other than MSB) occurring in the binary representation of the wavelet coefficients.\\
\item \emph{Symbol $p$}:  Symbol $p$ denotes peak location. It is an additional symbol used to identify the location of peaks . It can be viewed as a $flag$ bit in the encoding process that signifies presence of a peak. The location of wavelet coefficients which should be appended with $p$ is retrieved by computing energy packing density in the lowest level wavelet coefficients. Symbol $p$ can be appended at the beginning or end of the respective encoded value.\\
\end{enumerate}


\begin{figure}[h!]
\centering
\includegraphics[scale=.5]{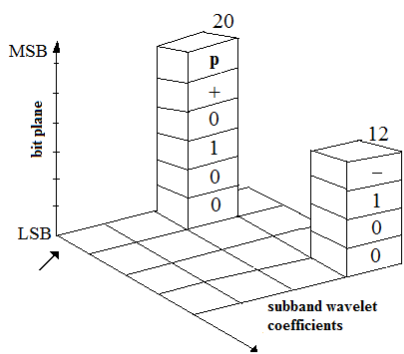}
\caption{symbol Stack Run Coding in wavelet domain}
\label{fig:src}
\end{figure}

For instance, in Fig.~\ref{fig:src}, the arrow points to direction of raster scan. We first come across two non significant values, before we reach our first significant value $20$. Hence, the binary value for $2$ run length is $10$ and equivalently it is coded as $+ -$ as discussed earlier. The binary value corresponding to $20$ is $10100$ and its corresponding coded value is $+0100$. Note that in the code $+$ represents both MSB and sign of the significant coefficient. Since, it was recognized as a peak location, it is appended by symbol $p$ before MSB. Intuitively, peaks are positive; Hence, we can use $p$ to replace $+$ to represent  both MSB and peak. That is, $p$ is equivalent to $p+$, and this is respectively taken into account during decoding. When raster scan comes across coefficient $-12$, its correspondingly encoded as $-100$. Where $1100$ was its binary value and  it has a negative sign, so the MSB is encoded as $-$. Unlike traditional stack run coding, in our approach, we reorder encoded data from MSB to LSB and do not increment LSB to represent the next binary value. Although, we introduce another symbol $p$ to represent peaks, we omit $+$ and implicitly take into account both the sign and MSB bit. Thereby, reducing any additional overhead involved due to the inclusion of a new symbol as a part of traditional stack run coding.
\\\\
During decoding, the reverse process carried out to decode the data. A raster scan is performed on the incoming symbol stream, which is equivalently mapped to its corresponding binary code of the encoded wavelet coefficients. Since the symbol $p$ is only appended on the identified peak wavelet coefficients from the lowest level, encoding an additional symbol does not affect compression performance. And this peak information is used to encode the recovered MT signal based on its $3$ transition states, and to estimate the MT parameters. Thus, in this work, we have employed an effective and efficient wavelet based compression technique to uniquely encode the peak information derived from wavelet domain, as well as provide additional advantage of compression of MT signals.

\subsection{Peak Detection Using Wavelets}
\label{sec:peak}
Wavelets not only offers desirable qualities like, sparsity and inherent noise reduction, but also allows simultaneous time-frequency resolution. Higher wavelet levels  provide better frequency resolution and the lower levels provide better time resolution; Unlike Fourier transforms, where we have either time or frequency resolution at any instant. Thus, this simultaneous time-frequency resolution feature of wavelets can be exploited in our case to locate the peaks present in the MTs. Peak detection  has a significant role in detecting the points of transition between the three states in MTs. We then use the transition frequencies and transition peaks extracted to estimate the MT parameters and understand the dynamic instability present in the system.  Peak detection in wavelet domain is performed based on the Energy Packing Efficiency (EPE )as in \cite{30}. Where EPE gives us a measure of total energy preserved in a specific subband after some thresholding condition is applied with respect to the total energy preserved in a subband before any thresholding and it is defined as given below \cite{30, 31}:
\be
\label{10}
EPE_i=\frac{E_{THi}}{E_{TOTi}} * 100
\ee
\\
where  for every sub band $i$, $EPE_i$ stands for Energy packing efficiency, $ E_{THi}$ stands for the energy preserved in subband after thresholding, $E_{TOTi}$ stands for the total energy preserved in subband before thresholding.

Since we are interested in the time instances where switching occurs (transition frequencies), we sort the significant wavelet coefficients in the lowest  level in ascending order to exploit time resolution property of wavelets. This is because the lowest level has narrowest bins,  which aids in locating the peaks in time domain of the MT signal. A  threshold is then applied to retain a specific percentage of the significant coefficients, where each of the retained coefficients correspond to peaks in the time domain. The location of those significant coefficients is identified with a $flag$. Later on, this $flag$ is used by the stack run encoded to identify the locations of significant wavelet coefficients that have to be appended with a symbol $p$ denoting peak.

\section{Results and Discussion for Experimental Data}
\label{sec:results}
In this section, we discuss the experimental results and the efficacy of our wavelet based stack run coding compression approach on MTs.  The MT data used for this work are results of experiments on MTs (composed of purified $\alpha \beta_{II}$ isotopes from bovine brain tubulin) performed by O. Azarenko, L. Wilson and M.A Jordan at the University of California, Santa Barbara. Tubulin proteins were first purified from the bovine brain and then seeded to polymerize at $37^{\circ}C$.  
Using the differential interference contrast video microscopy, every $2-6$ s time interval, data points were collected representing the length individual purified MTs  at their plus ends. 
MT lengths were analyzed using the Real Time Measurement program. Growing and shortening rates were calculated by least-squares regression analysis of the data points for each phase of growth or shortening. Growing and shortening thresholds are set to an increase in length by $0.2~ \mu$m at a rate of $0.15 ~\mu$m/min and a decrease in length by $0.2~ \mu$m at a rate of $0.3~ \mu$m/min, respectively. Any length changes equal to or less than $0.2~ \mu$m over the duration of six data points were considered attenuation phases  since the length changes were below the resolution of the microscope. It should be noted that the experimental detection limit for length changes corresponds to about $400-800$ tubulin dimers. The supplied data, however, was in the form of a hard copy graphs, that they were scanned and then digitized using the software "DigitizeIt" (\url{http: //www.digitizeit.de/}).  
See Rezania et. al for more details \cite{32}.

\begin{figure}[tp]
     \centerline{\epsfig{width=2.35\figurewidth,file=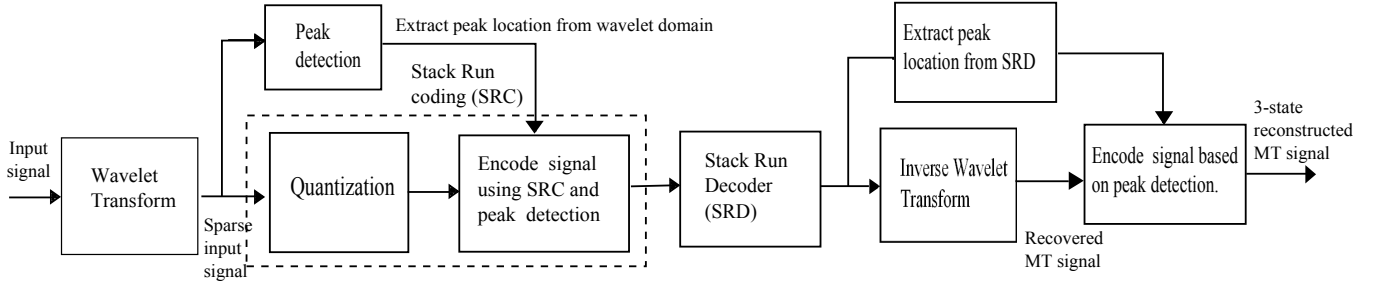}}
     \caption{Framework of our proposed wavelet based stack run coding compression of MT signals}
     \label{fig:block_diagram}
\end{figure}

In our proposed method, we first apply discrete wavelet transform to get a sparse representation for MT length signal. Fig. \ref{fig:ori_sig} illustrates one of the MT length datasets used for our experiments. In this work, we have used Daubechies D2(db2) mother wavelet with $8$ levels of signal decomposition, which  gave  better performance rather than other mother wavelet families and decomposition levels in our work. The sparse wavelet coefficients are then quantized by employing an uniform quantizer with a step size $q$ across all levels or subbands. We also perform peak detection on wavelet on the lowest level wavelet coefficients using energy packing density (EPE). The identified peaks are provided as input to stack run encoder during encoding process to encode the peaks with an additional symbol $p$ denoting peaks. The encoded stream is transmitted through the channel and decoded at the receiver side decoder. The symbol stream is mapped back to its wavelet coefficients and retrieve the information about location of peaks. Inverse wavelet transform is then performed on the MTs to recover the MT signal. The retrieved peak information is used to encode the resultant $3$ states of the MT signal and extract the switching frequencies between the three states of growth, shrinkage and pause. Fig~\ref{fig:block_diagram} elucidates the framework of the proposed approach.
\\\\
The effects of using our proposed compression scheme on dynamic instability of the MTs are studied by estimating various MT parameters for varying quantization step size $q$.  Lower the step size $q$, better the approximation of the signal, and lower the data compression. Conversely, higher the step size $q$, lesser the approximation of the signal, and higher data compression. To measure data compression, we define compression ratio (CR) as below:
\\\\
\be
CR= \frac{\text{Number of bits in the original data}}{\text{Number of bits in the compressed data}}
\ee
\\
Fig. \ref{fig:src_recovery}, depicts the reconstructed $3$-state MT signal for varying values of $q$. Where the $blue$ signal is the original encoded $3$-state signal and the $red$ signal is the reconstructed $3$-state signal. We can confirm from Fig. \ref{fig:src_recovery}, that for lower rates of compression (smaller $q$), we have better approximation of the signal and for higher rates of compression (larger $q$) we tend to over approximate or under approximate the signal causing deviation of the predicted signal from original signal.
\\
\begin{figure}[tp]
     \centerline{\epsfig{width=\figurewidth,file=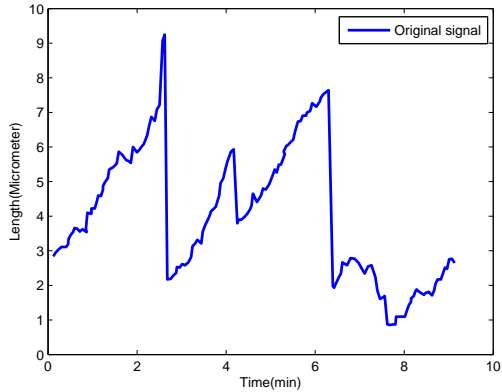}}
     \caption{Original signal}
     \label{fig:ori_sig}
\end{figure}

\begin{figure*}[tp] \small
\centering
\begin{tabular}{cc}
 \epsfig{width=0.95\figurewidth,file=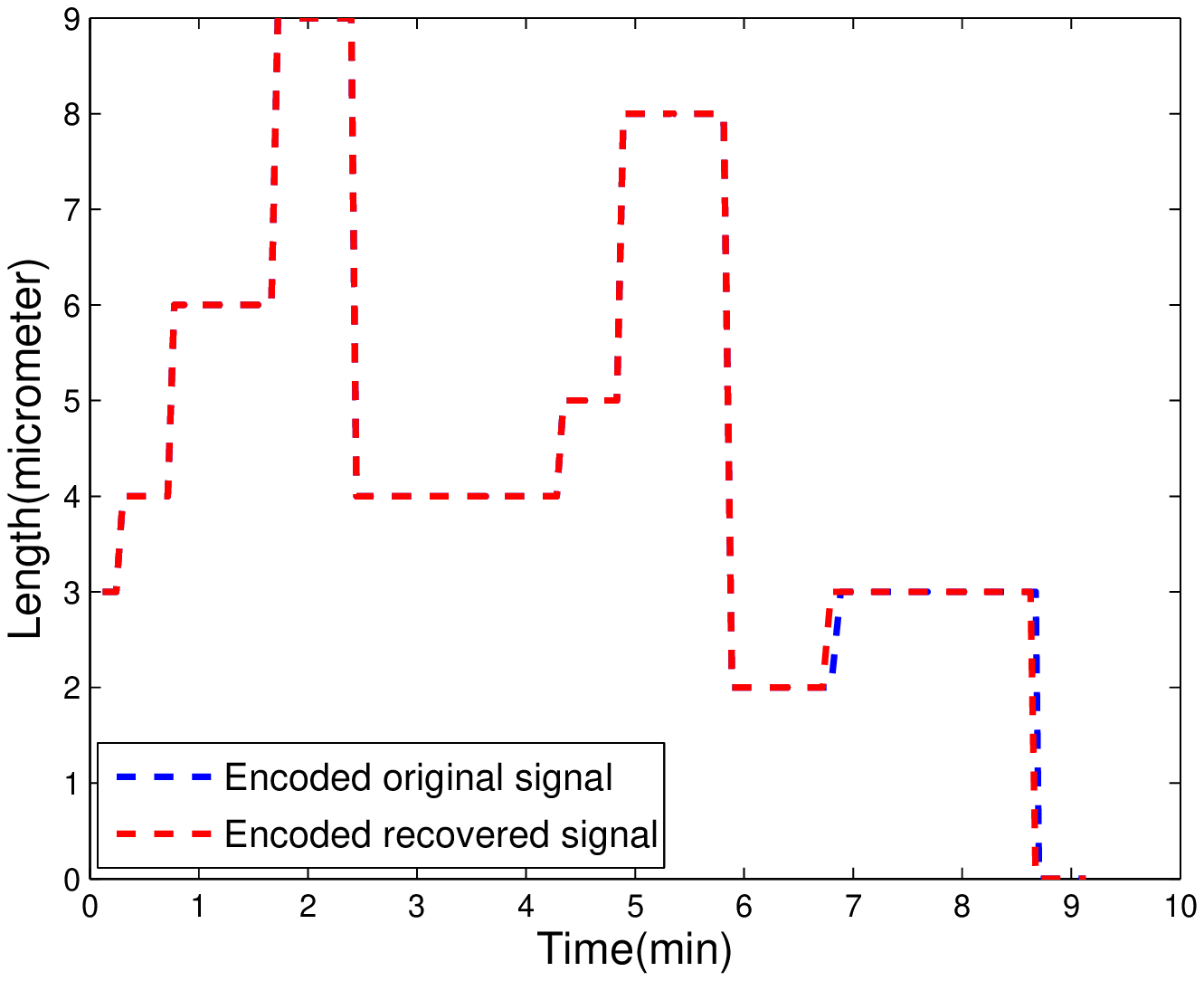}&
 \epsfig{width=0.95\figurewidth,file=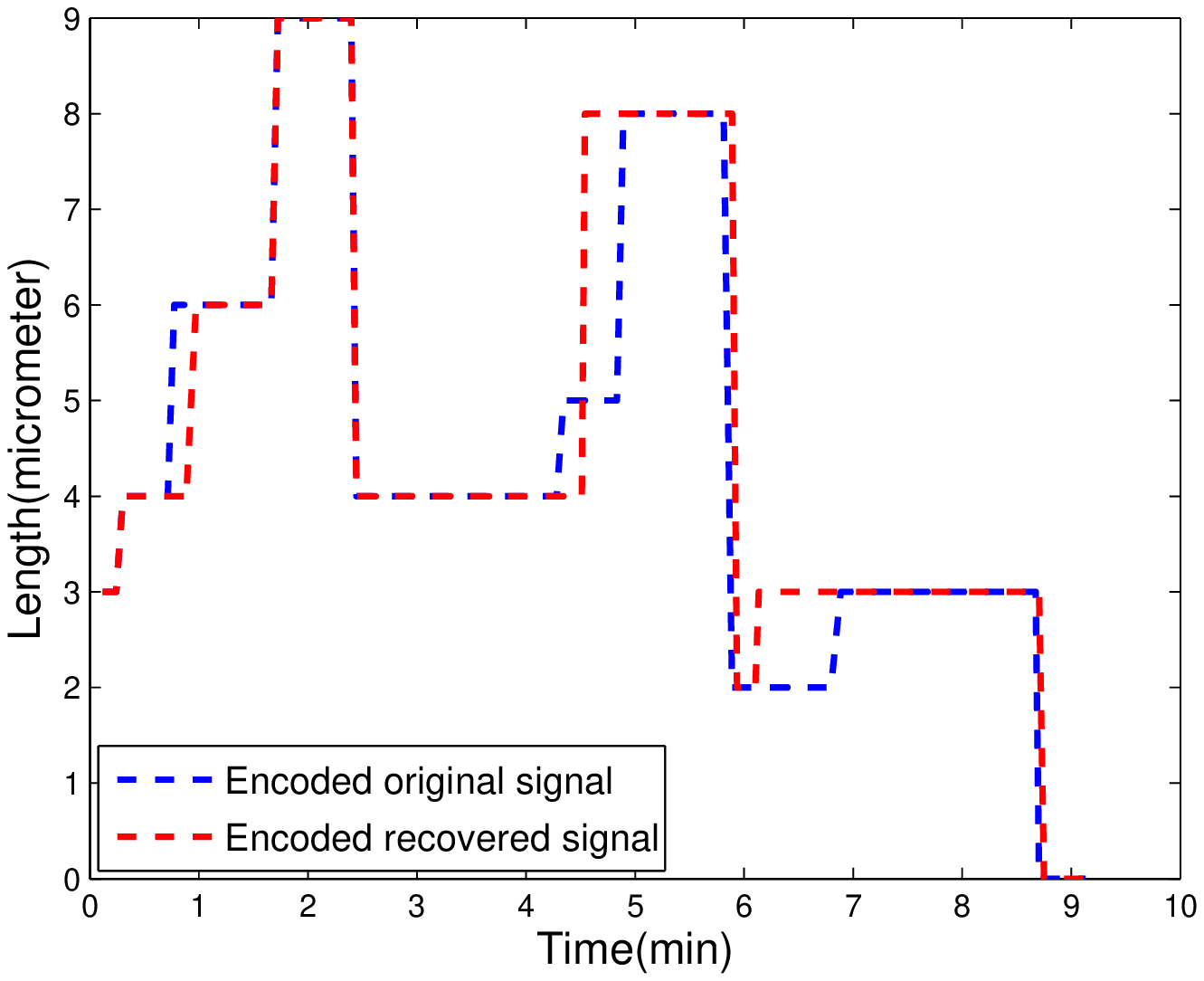}\\
     (a) Reconstructed MT signal for $q$=$0.04$ , $CR = 1.14$  &  (b) Reconstructed MT signal for $q$=$0.08$ , 
$CR= 1.33$   \\[0.5em]
 \epsfig{width=0.95\figurewidth,file=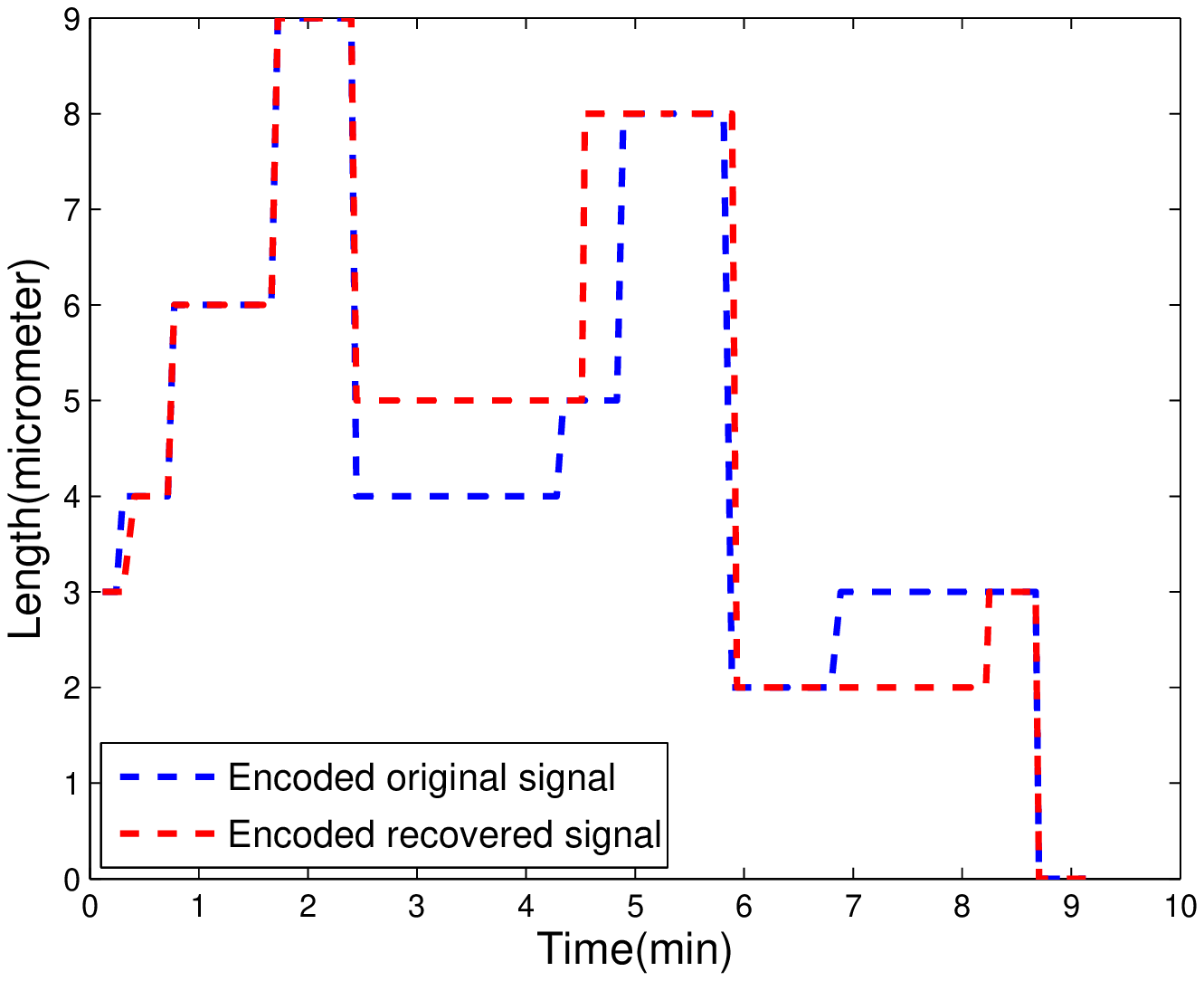}&
   \epsfig{width=0.95\figurewidth,file=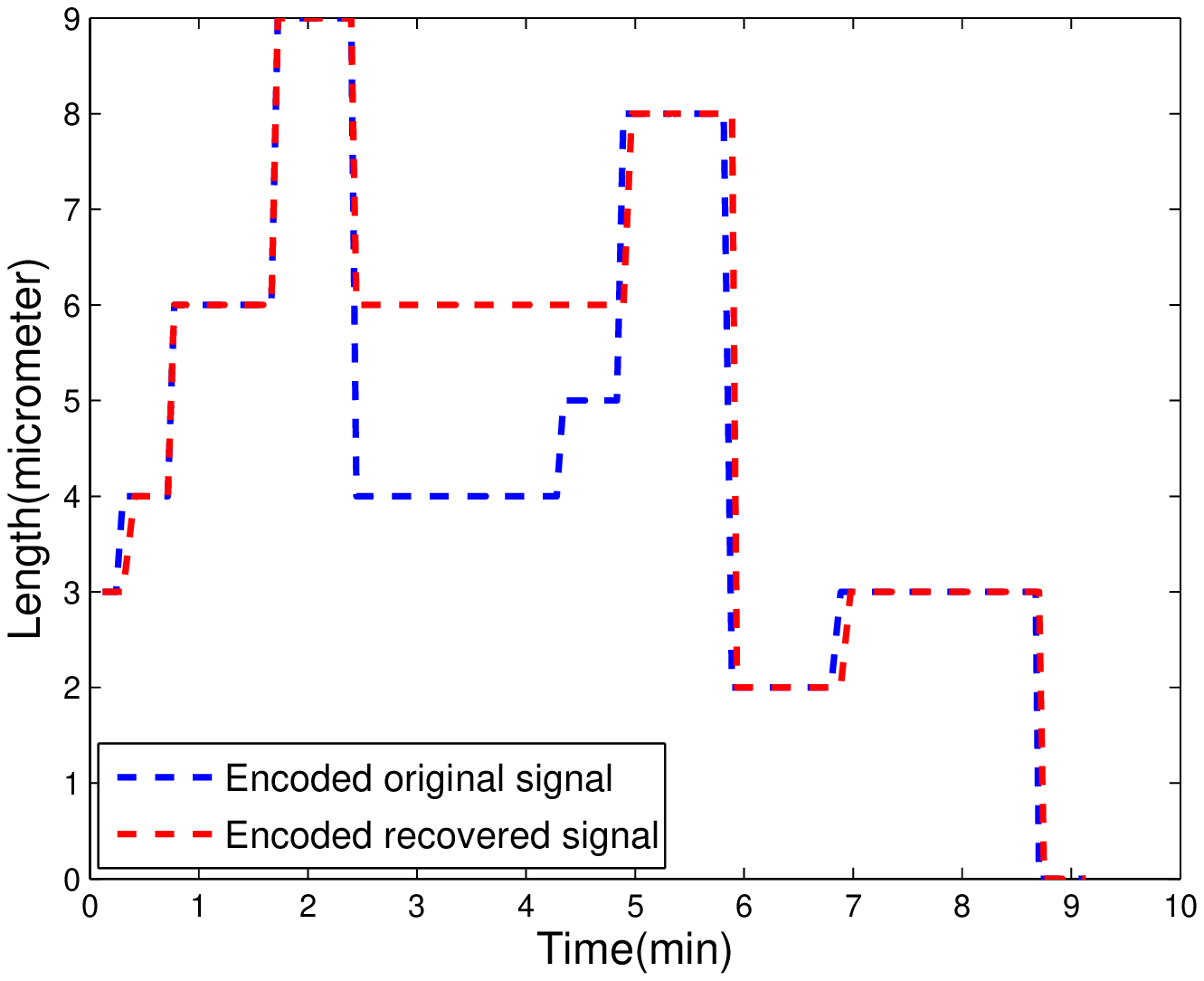}\\
        (c) Reconstructed MT signal for $q$=$0.12$, $CR= 1.52$ & (d) Reconstructed MT signal for $q$=$0.16$, 
$CR = 1.72$ \\[0.5em]
 \epsfig{width=0.95\figurewidth,file=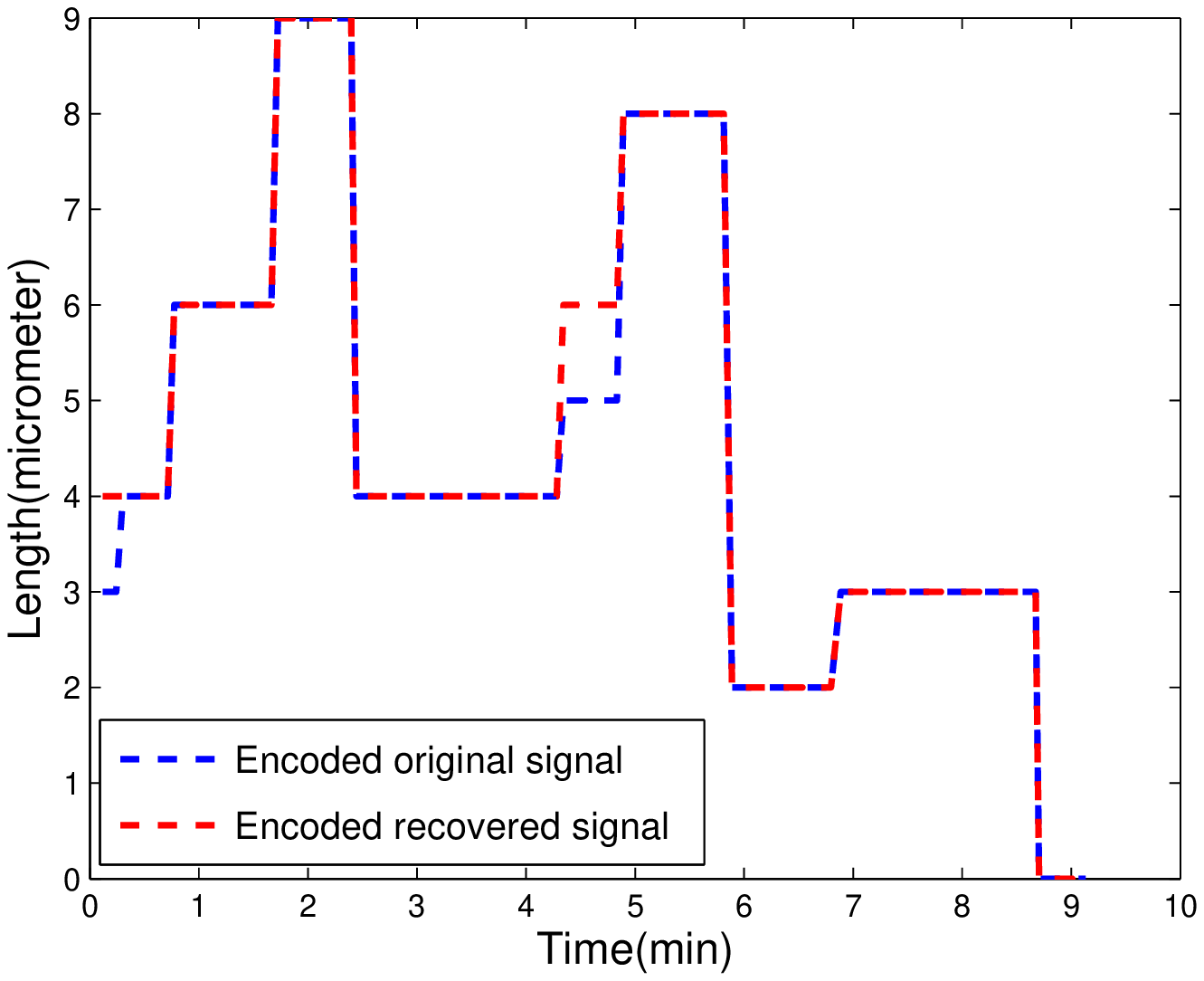}&
   \epsfig{width=0.95\figurewidth,file=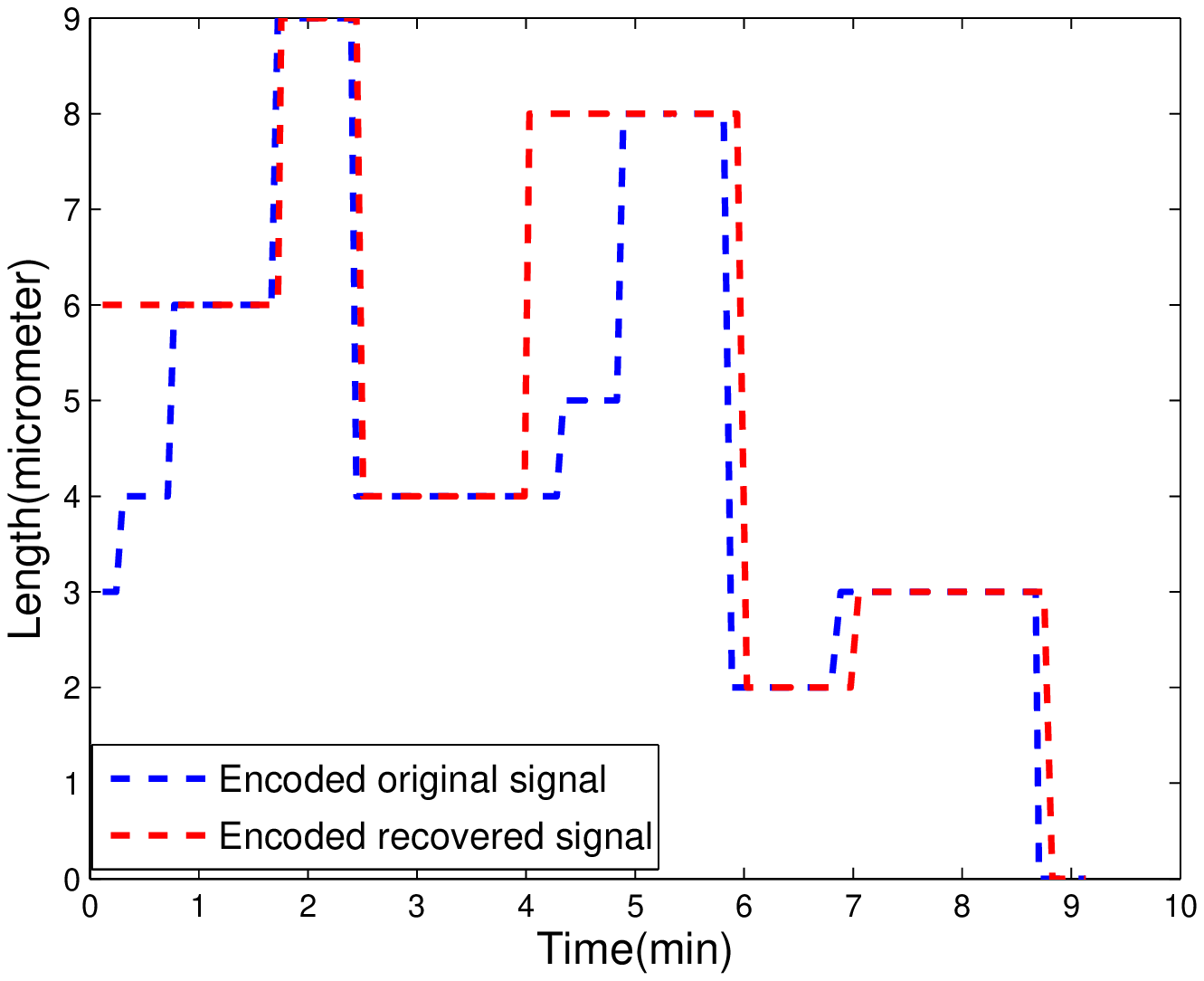}\\
(e) Reconstructed MT signal for $q$=$0.2$, $CR=1.96$ & (f) Reconstructed MT signal for $q$=$0.4$, 
$CR= 2.62$ \\[0.5em]
  \end{tabular}
\vspace*{-0.1in}
\caption{\label{fig:src_recovery}
Reconstructed 3-state MT signal for various quantization step size $q$.}
\end{figure*}
 
\begin{figure*}[tp] \small
\centering
\begin{tabular}{cc}
 \epsfig{width=0.95\figurewidth,file=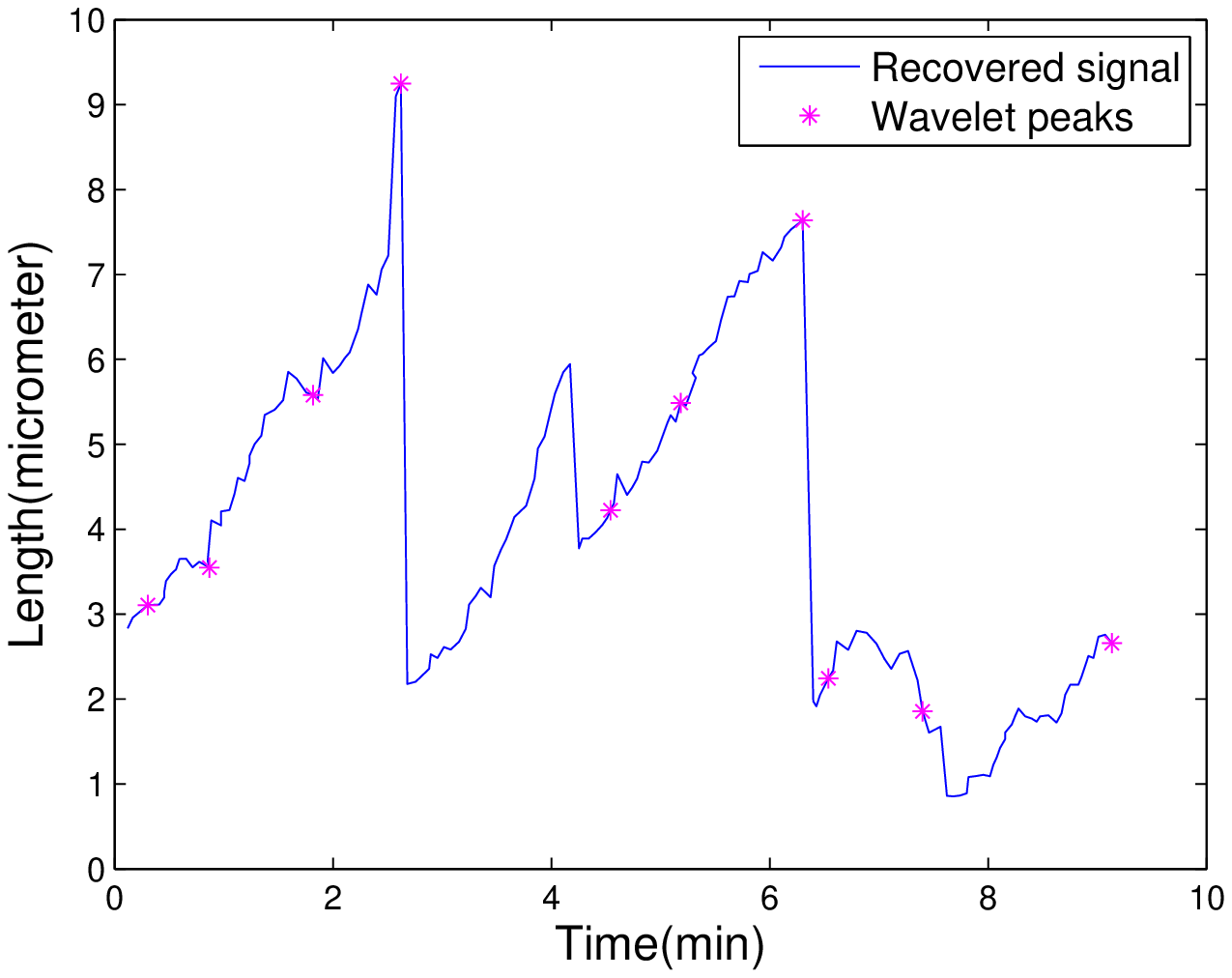}&
 \epsfig{width=0.95\figurewidth,file=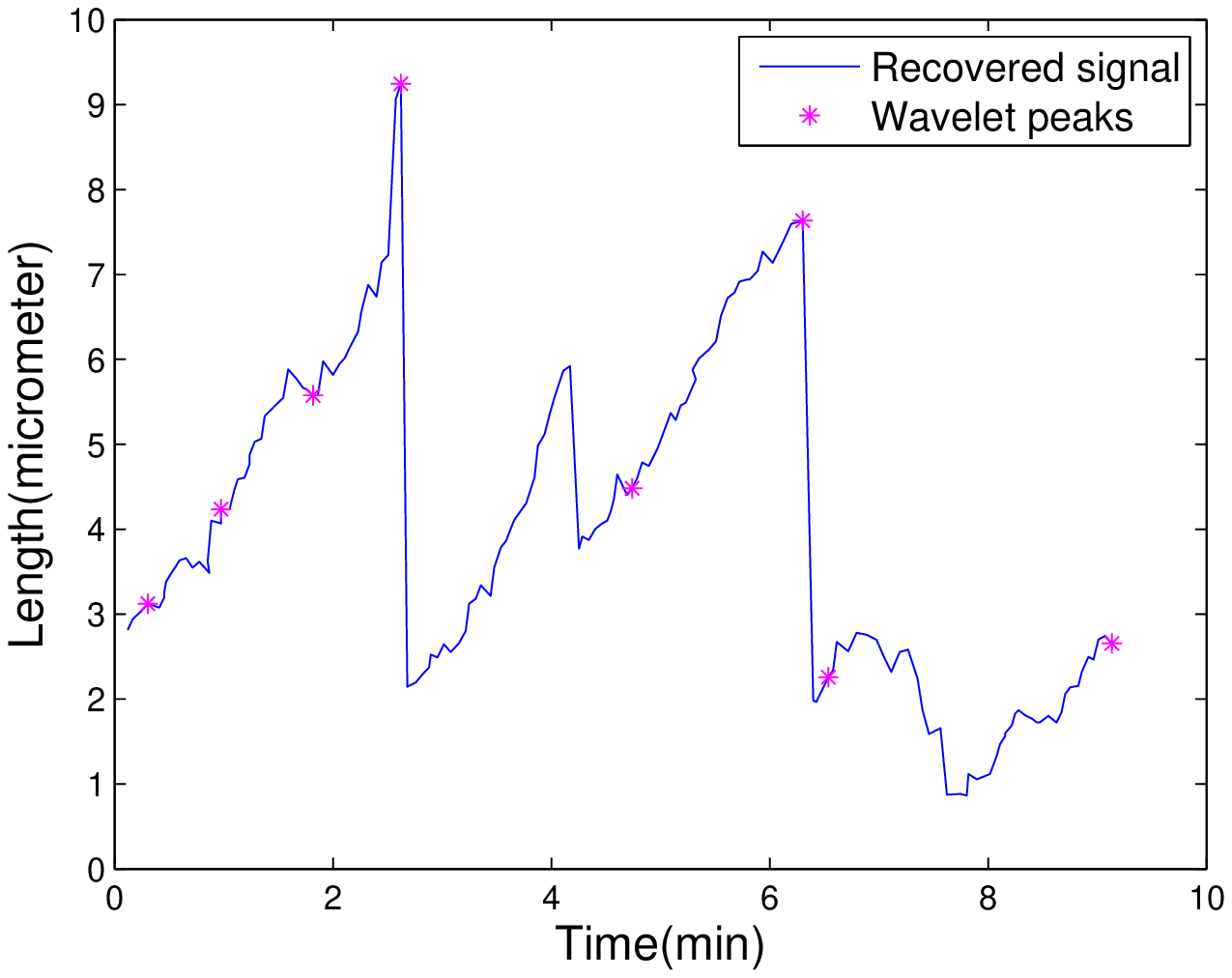}\\
     (a) Recovered MT signal for $q$=$0.04$ , $CR = 1.14$  &  (b) Recovered MT signal for $q$=$0.08$, 
$CR=1.33.$\\[0.5em]
 \epsfig{width=0.95\figurewidth,file=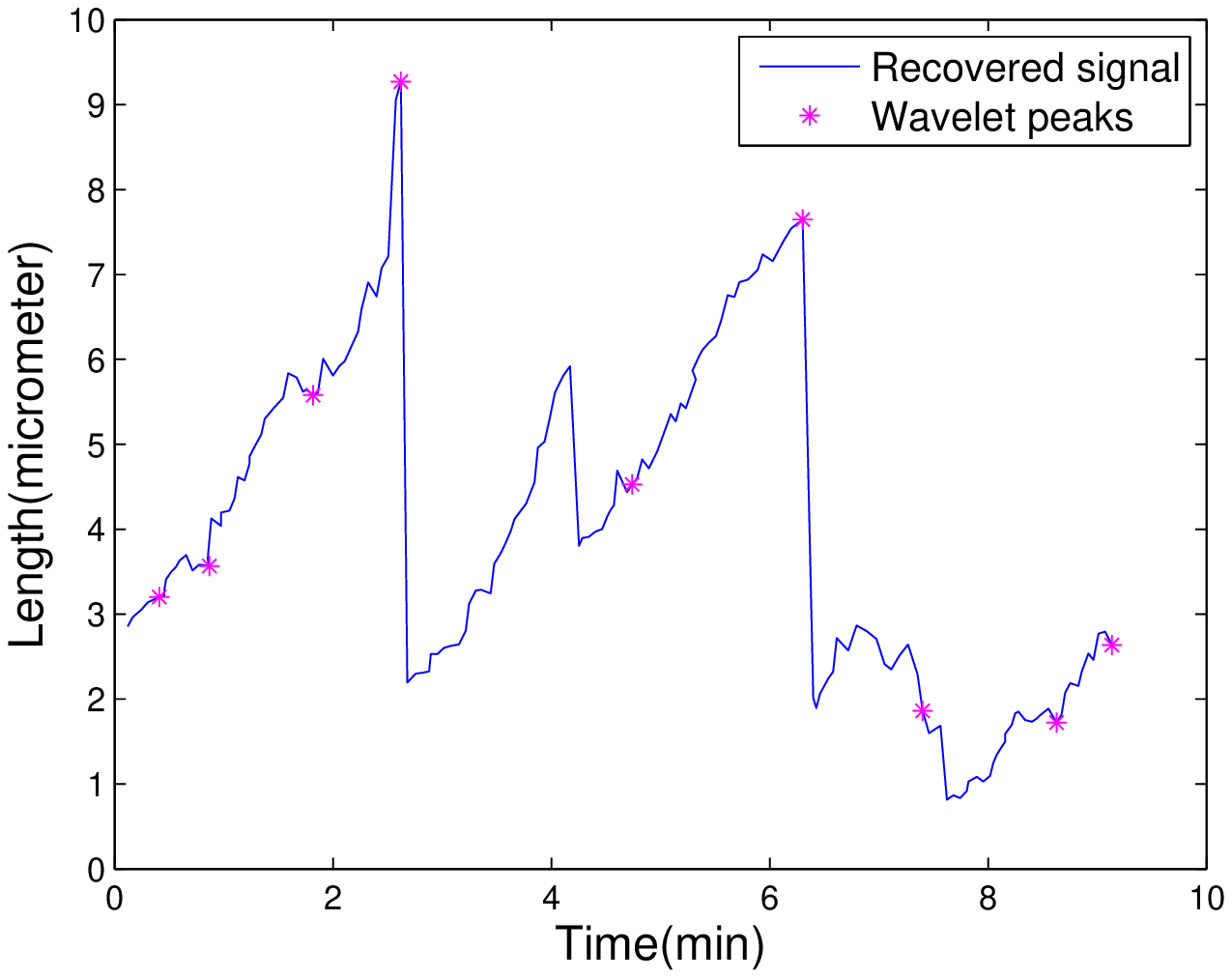}&
   \epsfig{width=0.95\figurewidth,file=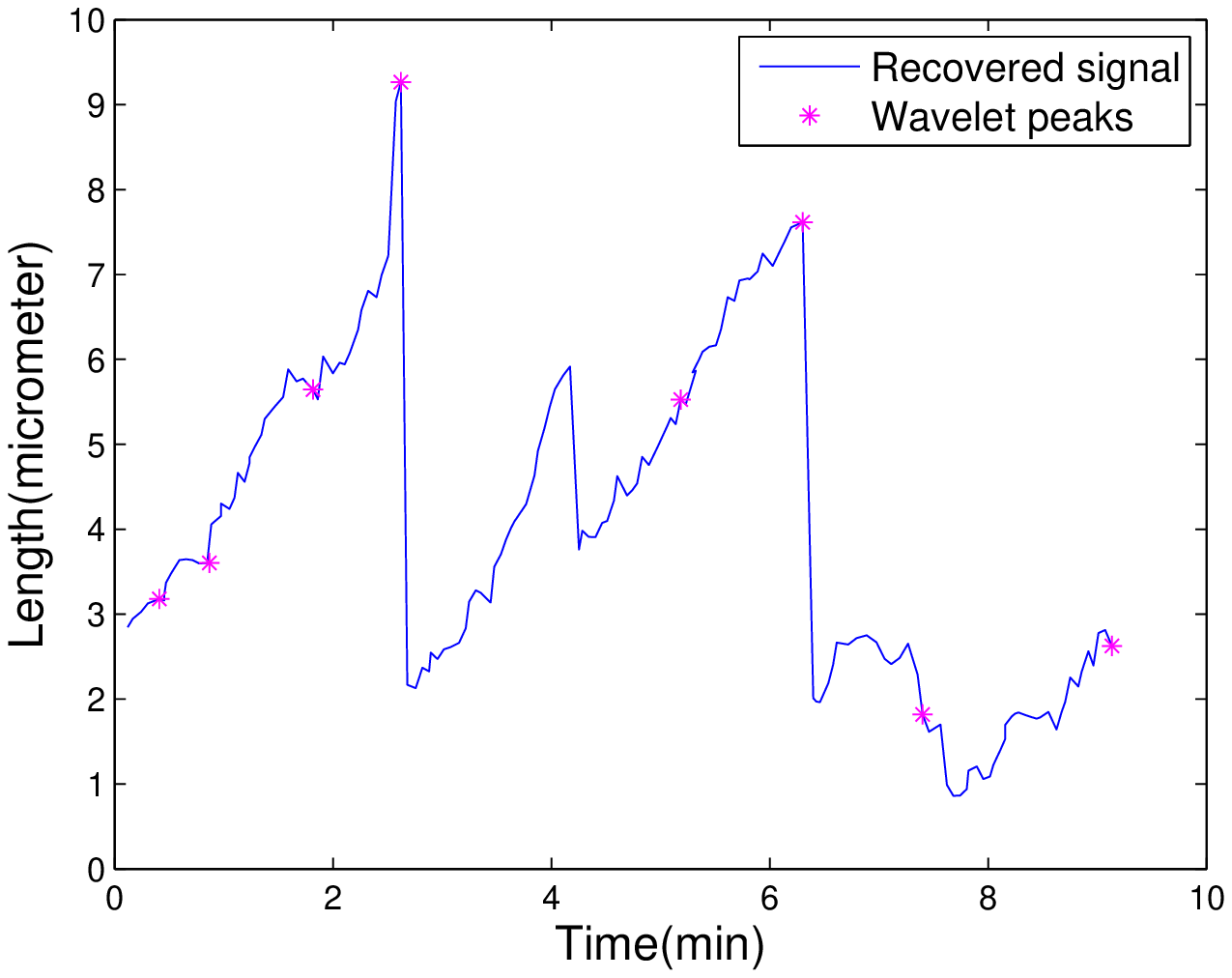}\\
        (c) Recovered MT signal for $q$=$0.12$, $CR= 1.52$ & (d) Recovered MT signal for $q$=$0.16$, 
$CR=1.72$ \\[0.5em]
 \epsfig{width=0.95\figurewidth,file=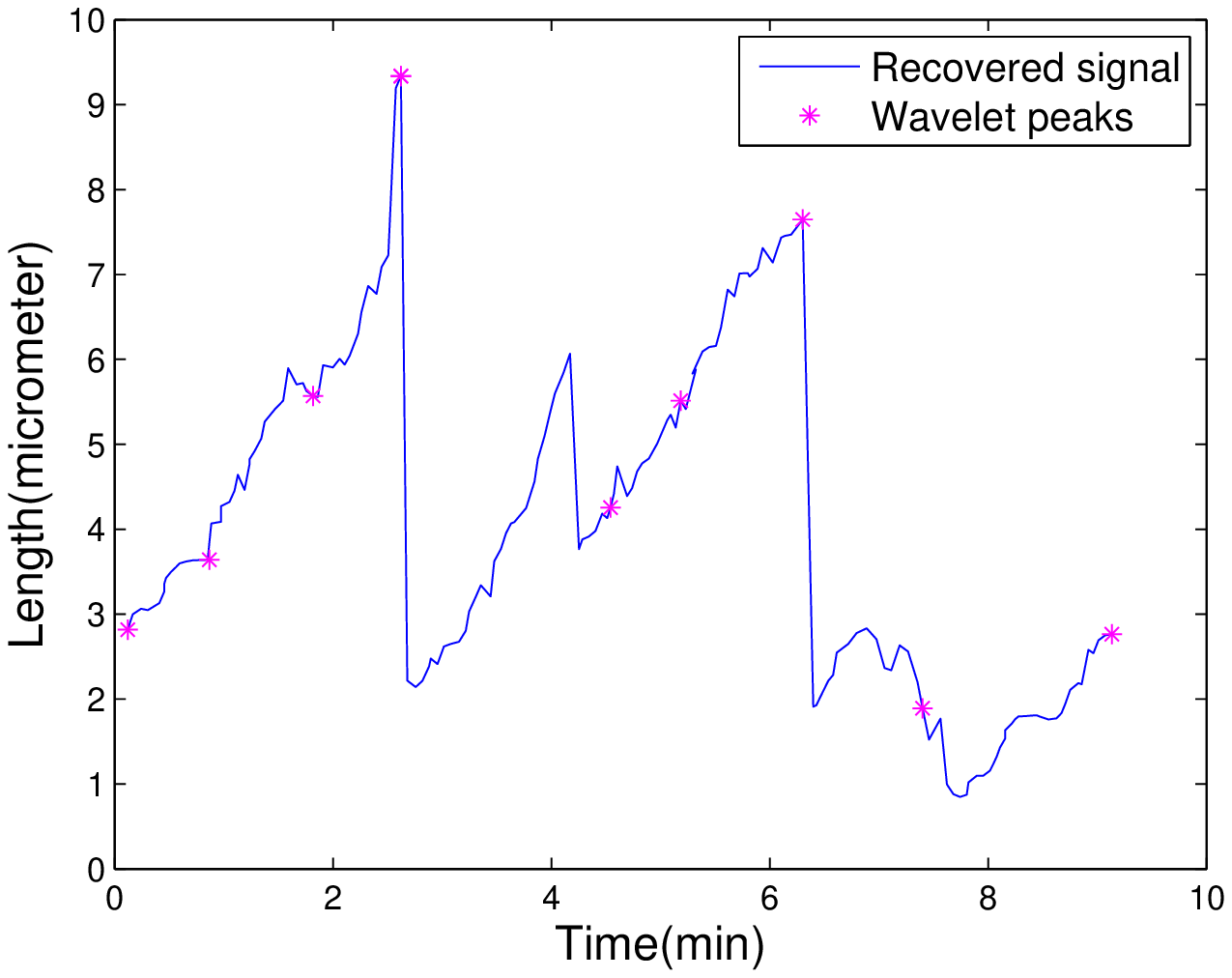}&
   \epsfig{width=0.95\figurewidth,file=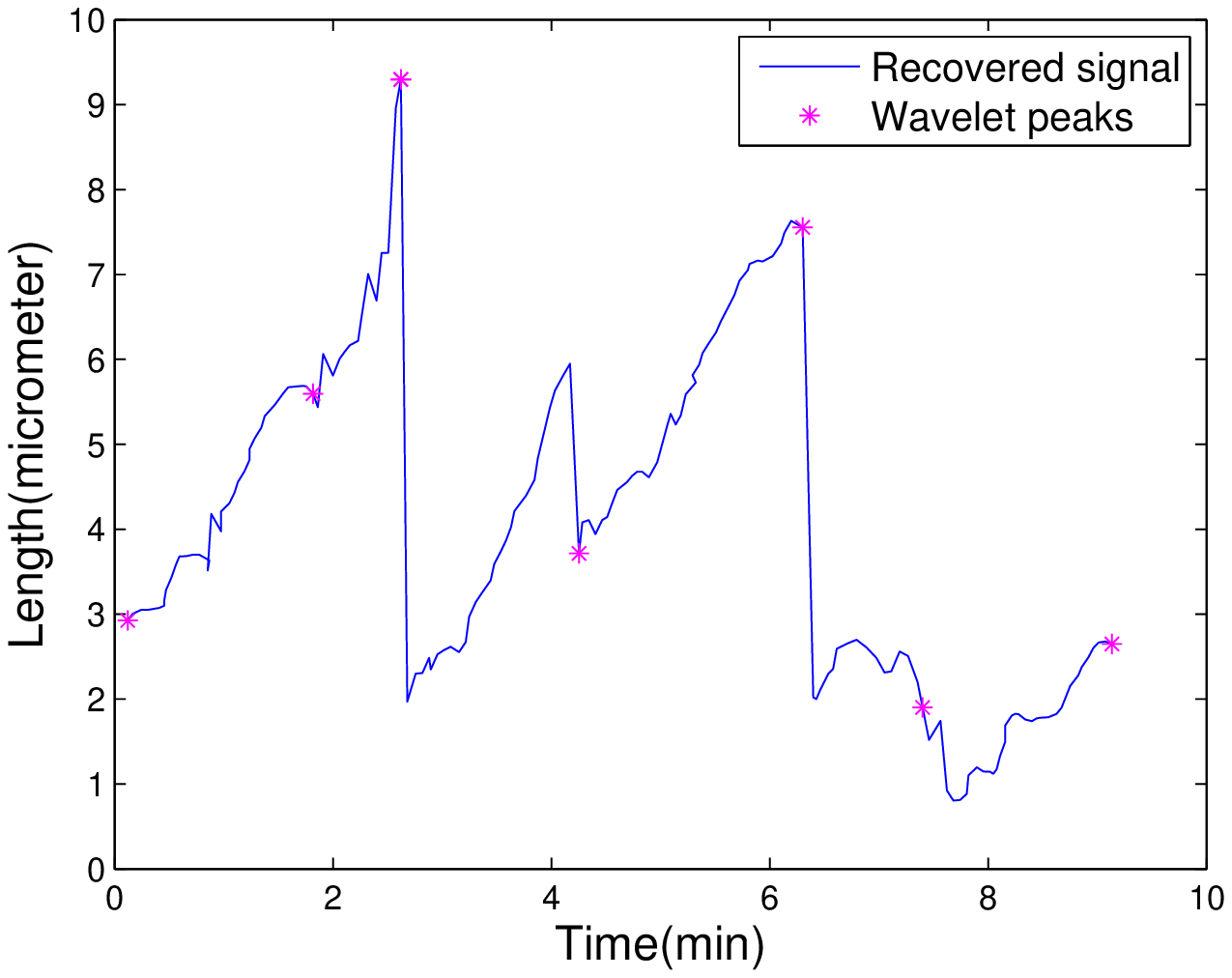}\\
(e) Recovered MT signal for $q$=$0.2$, $CR=1.96$ & (f) Recovered MT signal for $q$=$0.4$, 
$CR=2.62$ \\[0.5em]
  \end{tabular}
\vspace*{-0.1in}
\caption{\label{fig:wave_rec}
Recovered MT signal at decoder side for various quantization step size $q$.}
\end{figure*}

As discussed in our previous section \ref{sec:peak}, we perform peak detection in the wavelet domain, because of its inherent ability to reduce noise and simultaneous time-frequency resolution. It is essential to affirm the location of peaks occurring in MTs, because those peaks translate into points at which MTs transition between growth, shrinkage and pause states and to further estimate the MT parameters. Therefore, we employ EPE at the significant wavelet coefficients present at the lowest level, since they have narrow bins implying better time resolution. We then set a threshold EPE and select the qualifying significant wavelet coefficients after sorting them from highest to lowest values. We have set the threshold parameter as $E_{TH} \geq 0.85*E_{TOT}$. The thresholding procedure involved is as given below:

\begin{itemize}
\item Step 1: Sort the wavelet coefficients in the lowest level in descending order.
\item Step 2: Compute the  total energy present in the wavelet coefficients ($\bf{wc}$) in the desired level. Where total energy is calculated by :
\be
E_{TOT} = \sum{\bf{wc}^2}
\ee
\item Step 3: Fix the desired threshold percentage of significant wavelet coefficients that should be retained. In our case, we retain values such that $85\%$ of the total energy of the coefficients in that level is preserved. That is :
\be
E_{TH} \geq 0.85*E_{TOT}
\ee
\item Step 4: Sum up energy of the sorted coefficients until the threshold condition is reached.\\*
\end{itemize}
Thus, the number of significant wavelet coefficients selected is sensitive to our threshold parameter, since  it regulates the  number of potential peaks detected. Thereby, thresholding forms an important aspect in estimating the switching parameters of the TMN based MT signal. \\

The wavelet coefficients in the decoder side are recovered using inverse wavelet transform. The extracted peak information from the decoder side is used to round off the recovered MT length signal to its nearest peak value to obtain an encoded $3$ state MT signal.  We then study the changes due to dynamic instability of MTs occurring across various quantization step size $q$ by estimating MT parameters.  Fig \ref{fig:wave_rec} shows the resultant MT length signal after wavelet recovery for various step size values ($q$).
And its corresponding encoded $3$ state MT length signal is given in Fig \ref{fig:src_recovery}. From both the recovered and encoded MT signal figures, we note that for lower quantization step sizes, better peak detection, since we preserve more signal details, thereby providing better approximation of the encoded signal. For higher quantization step sizes,  we tend to lose details because of over approximation of signal. Therefore, this results in a smoother recovered signal and smoother transition between states causing  less accurate peak detection, hence leading to an over approximated MT signal. \\

\begin{figure*}[tp] \small
\centering
\begin{tabular}{cc}
\epsfig{width=1\figurewidth,file=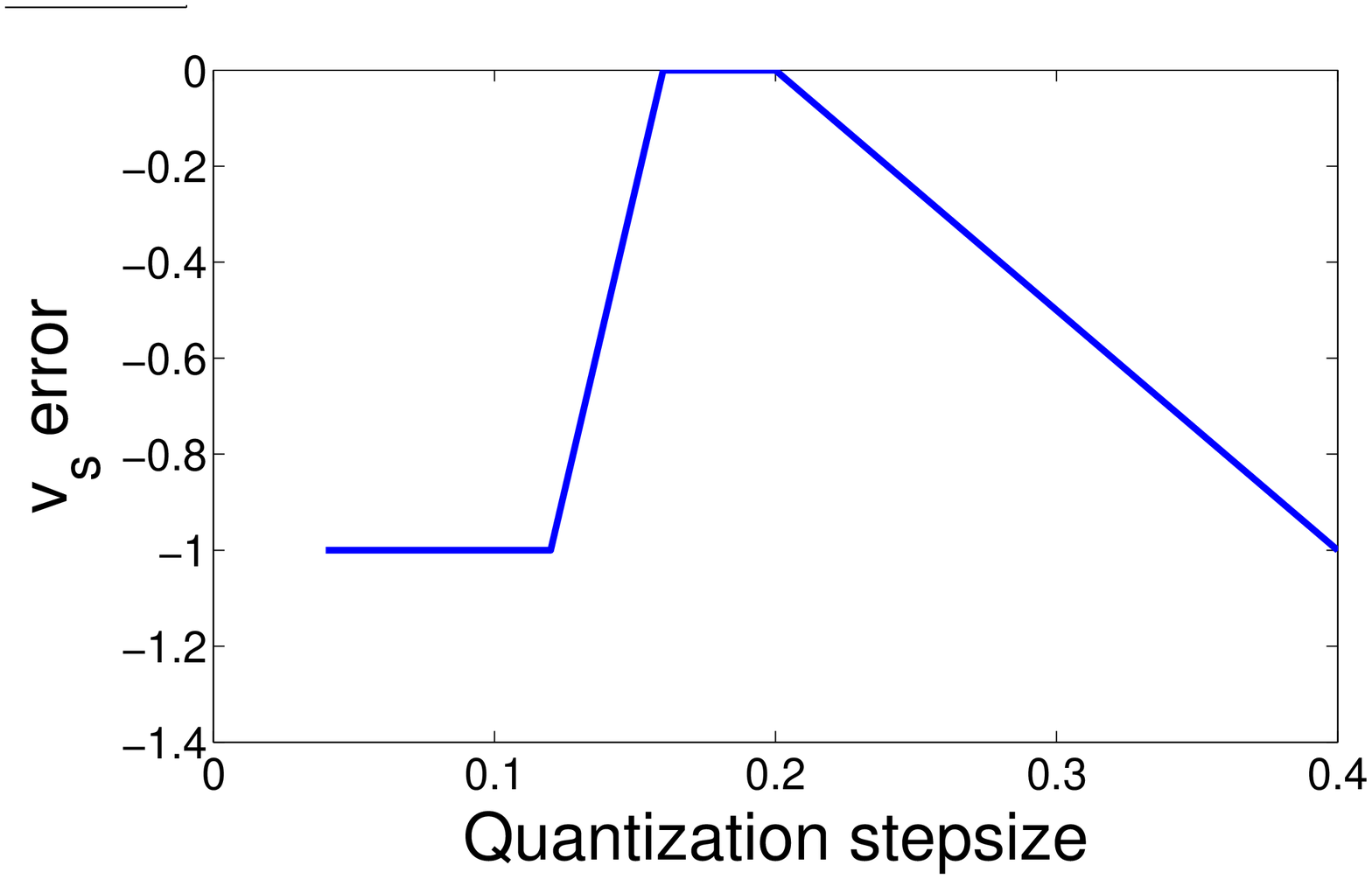}&
\epsfig{width=1\figurewidth,file=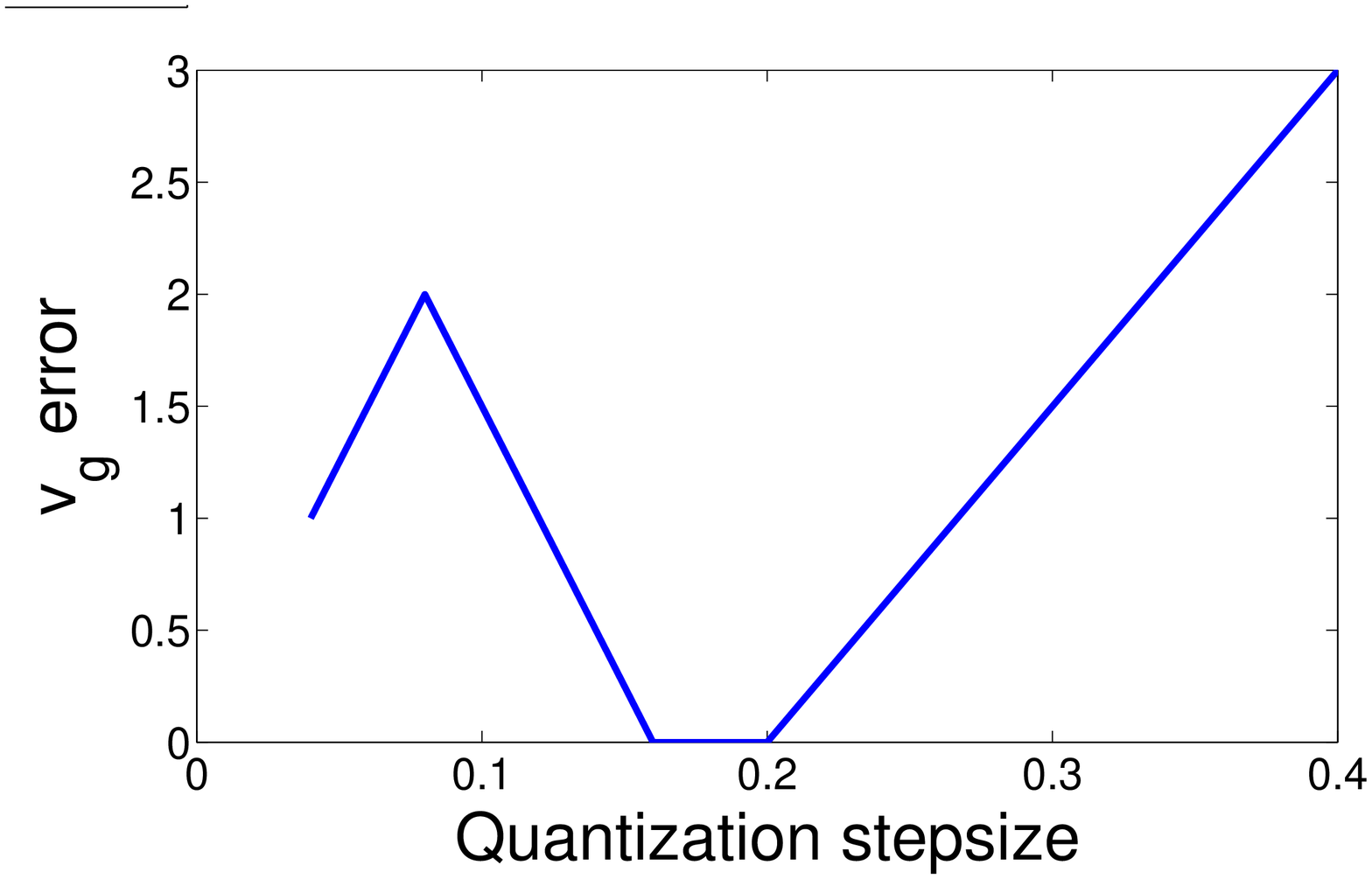}\\
      (a) Shrinkage rate error & (b) Growth rate error \\[0.5em]
 \epsfig{width=1\figurewidth,file=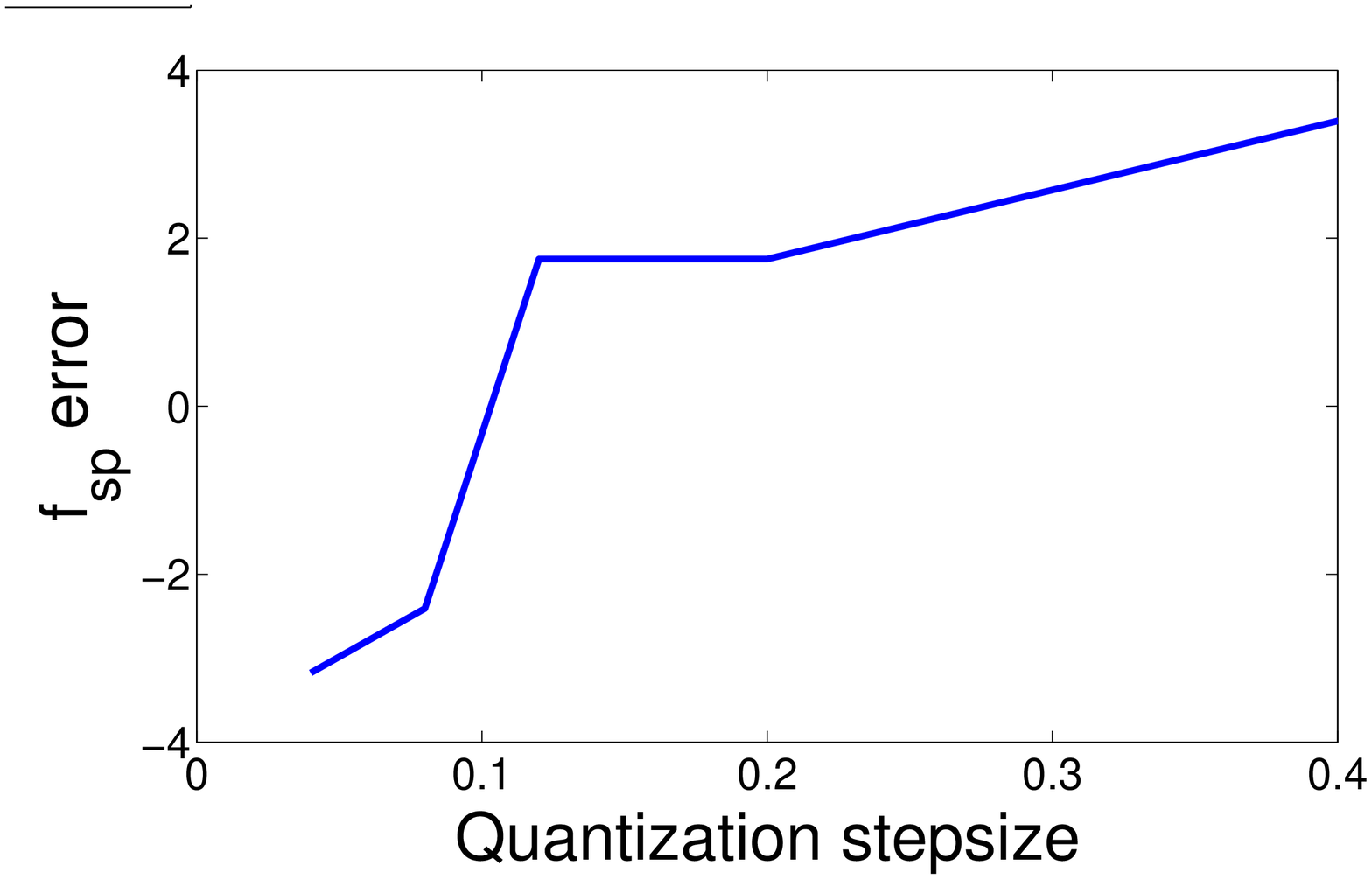}&
 \epsfig{width=1\figurewidth,file=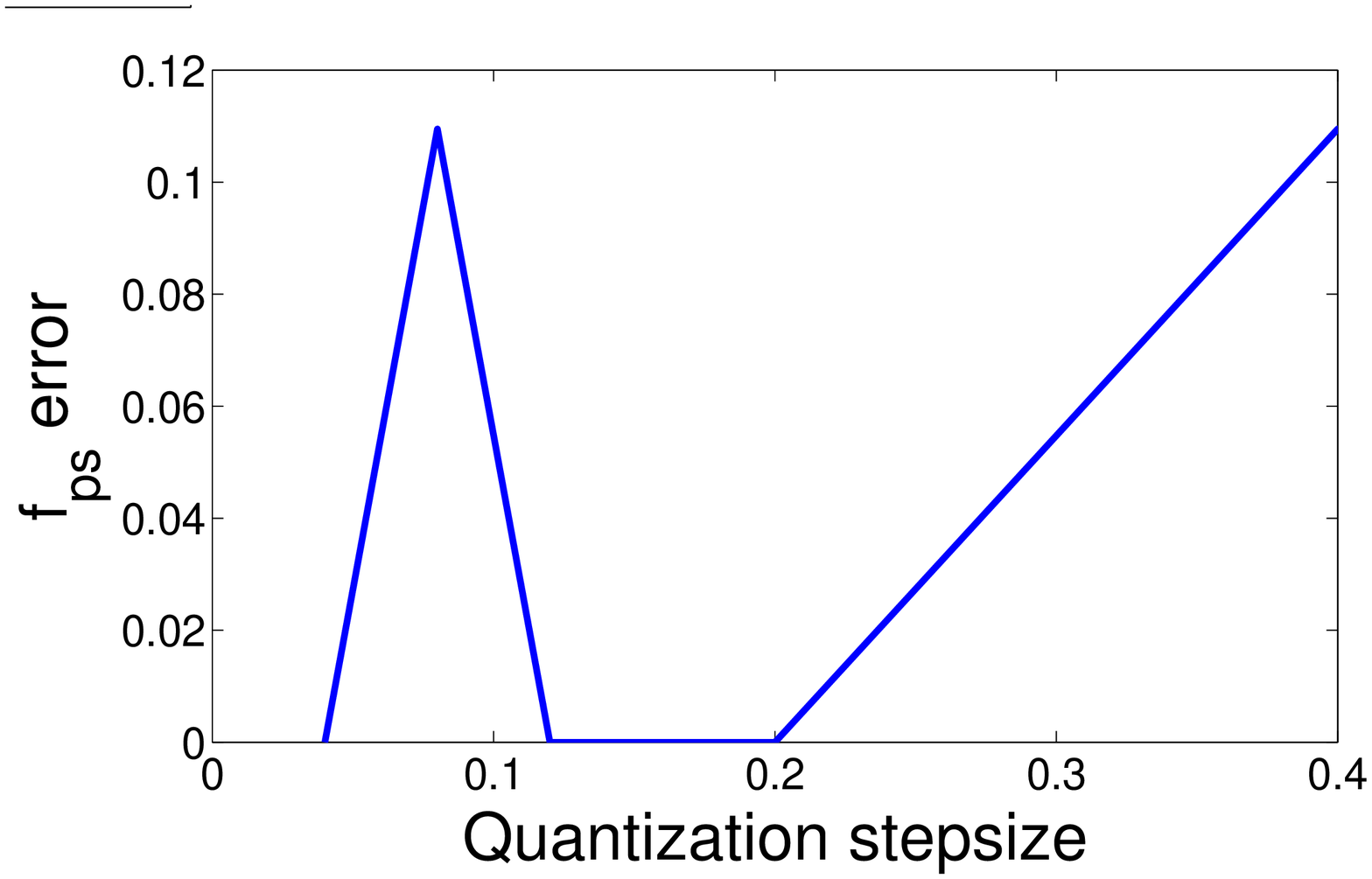}\\
        (c) Shrinkage to pause transition frequency error & (d) Pause to shrinkage transition frequency error  \\[0.5em]
 \epsfig{width=1\figurewidth,file=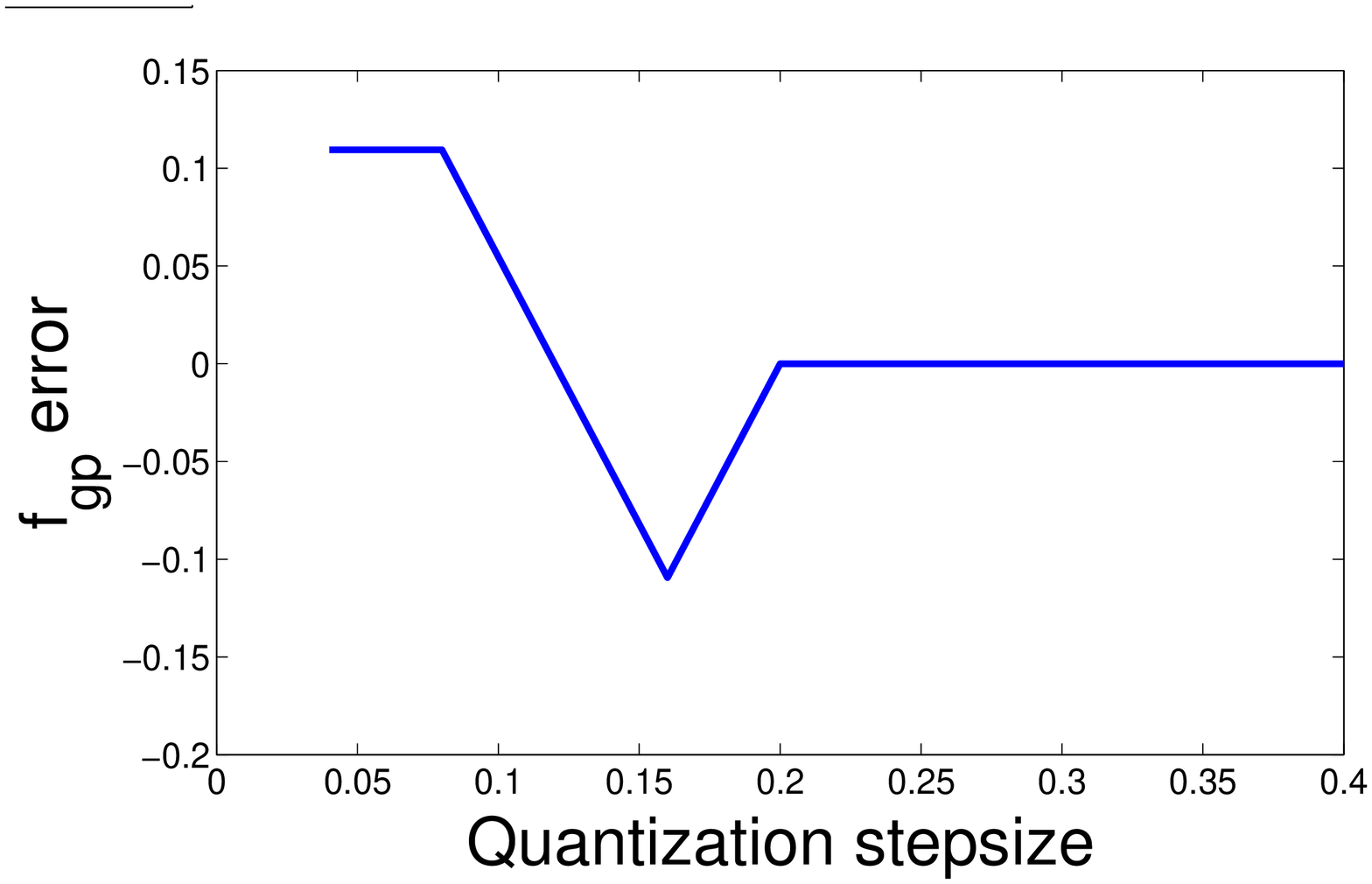}&
 \epsfig{width=1\figurewidth,file=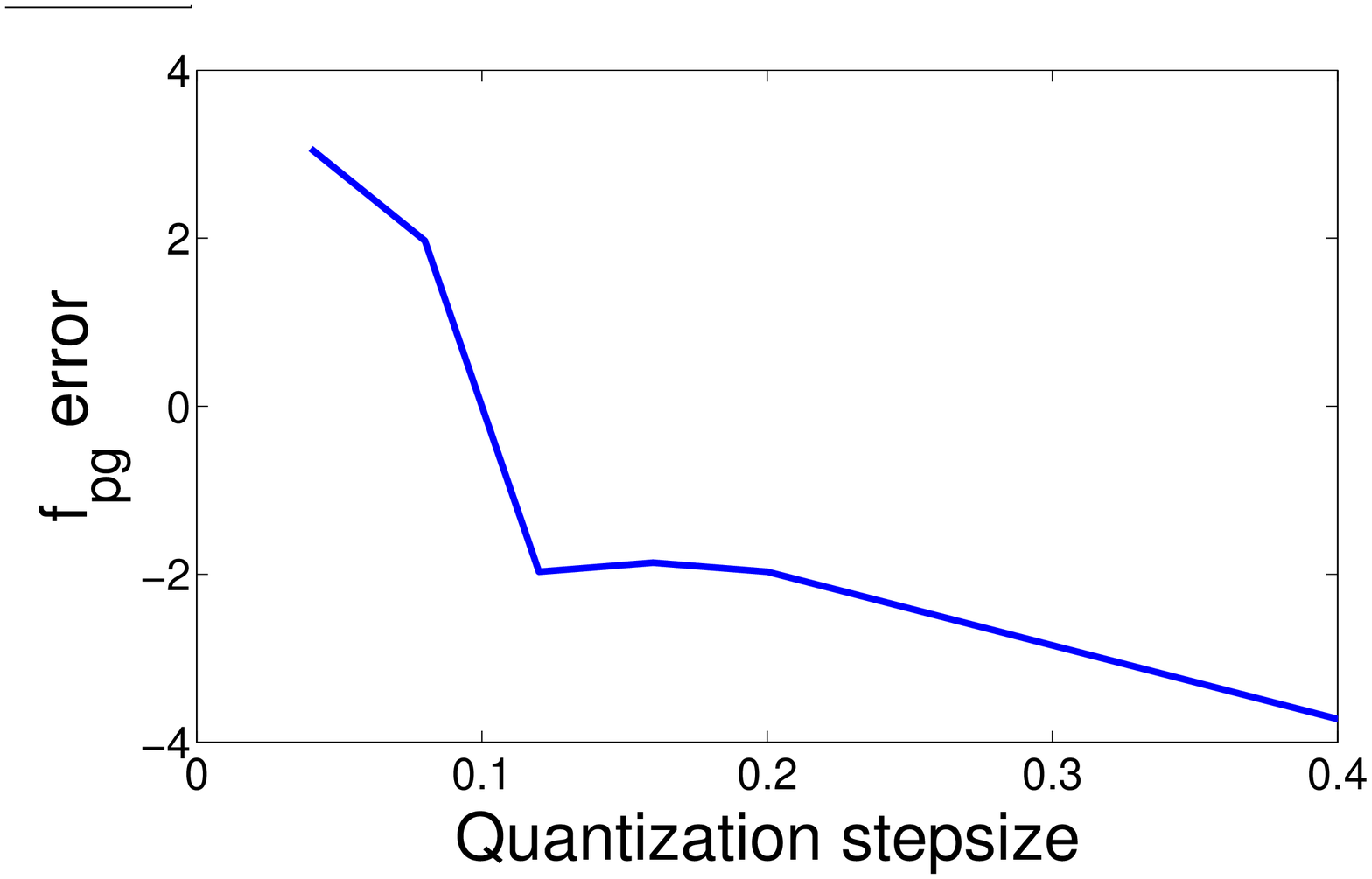}\\
        (e) Growth to pause transition frequency error  & (f) Pause to growth transition frequency error  \\[0.5em]
 \epsfig{width=1\figurewidth,file=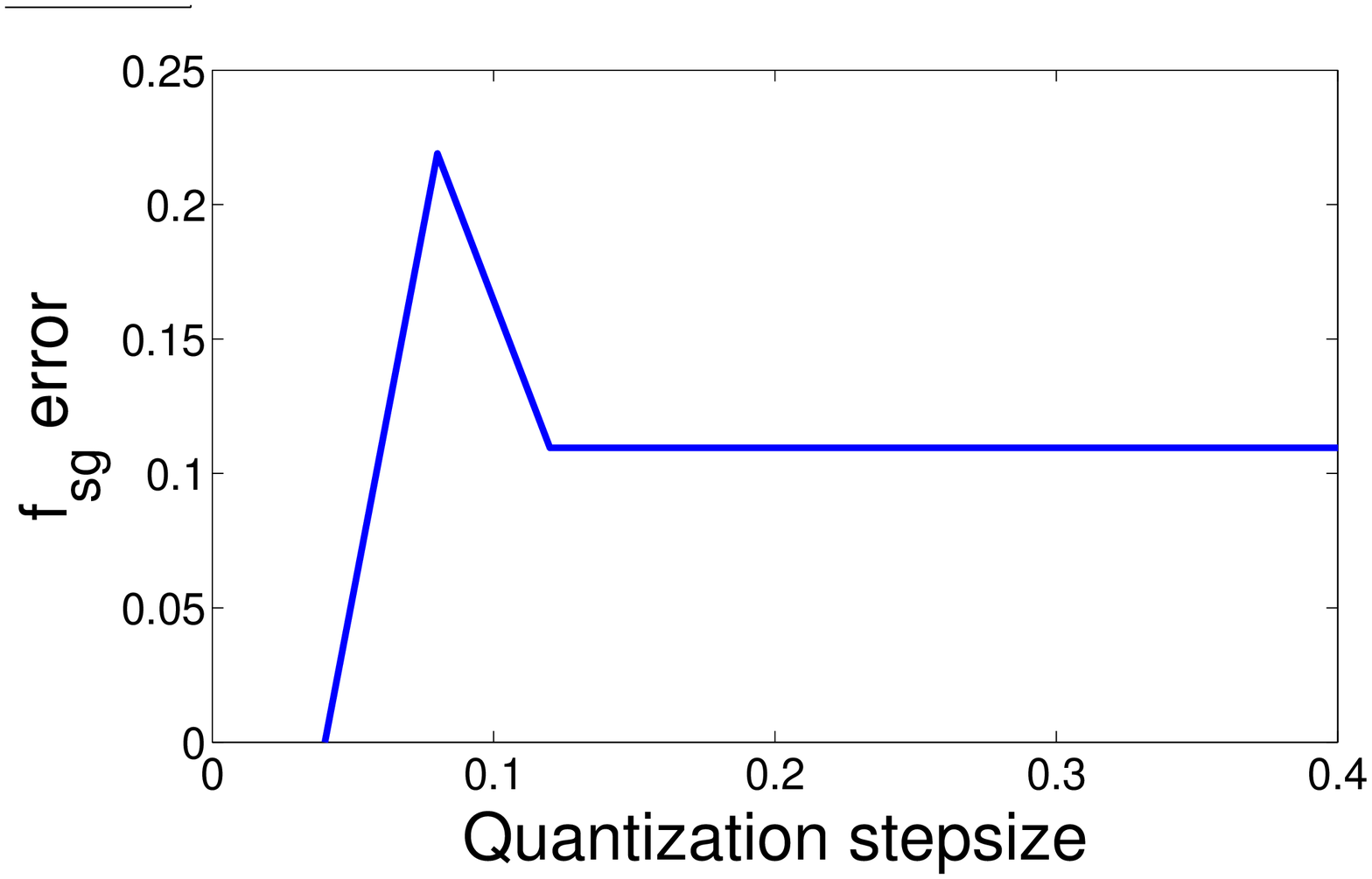}&
 \epsfig{width=1\figurewidth,file=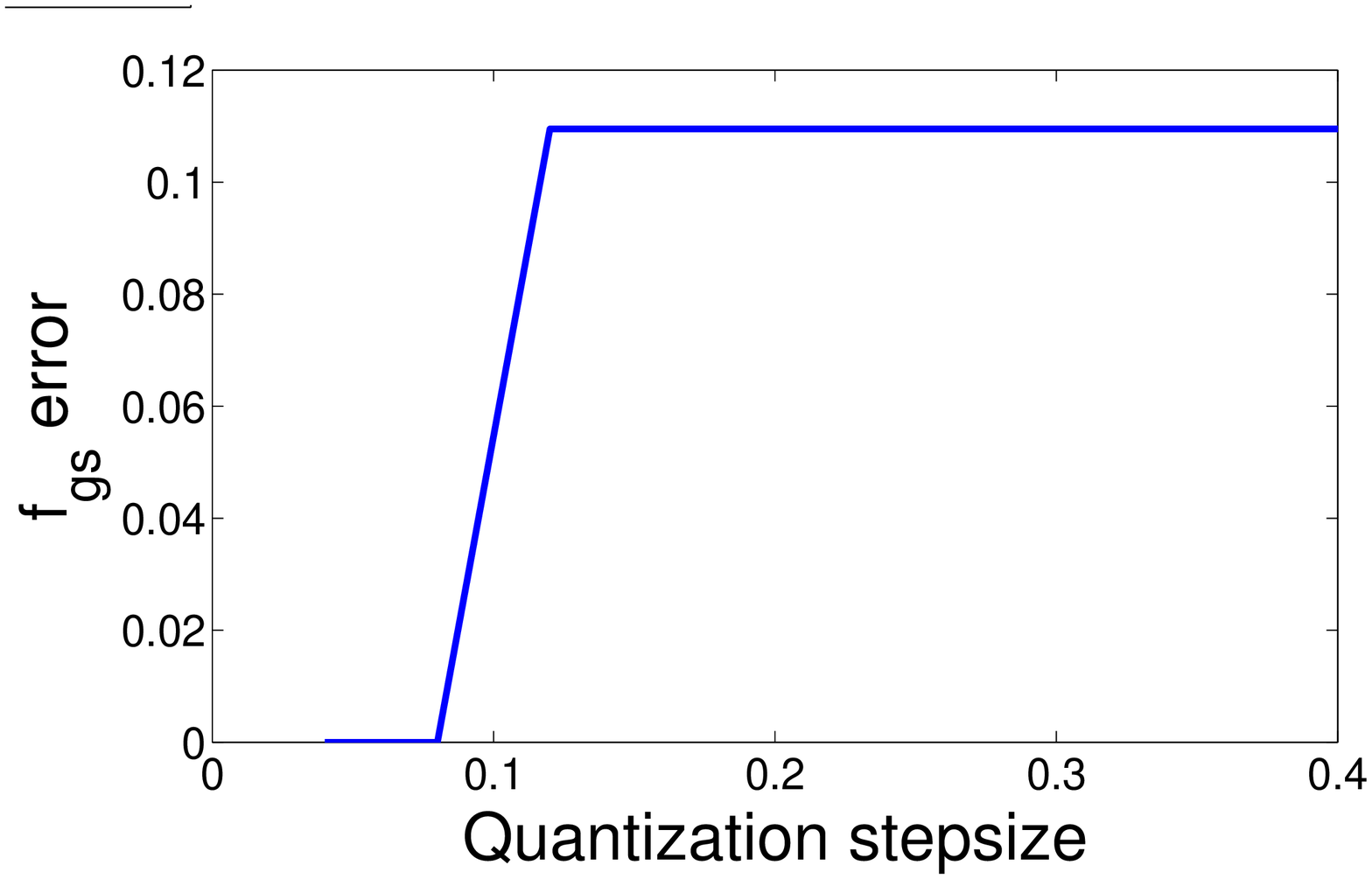}\\
        (g) Shrinkage to growth transition frequency error  & (h) Growth to shrinkage transition frequency error  \\[0.5em]
  \end{tabular}
\vspace*{-0.1in}
\caption{\label{fig:err}
Error estimation for MT parameters.}
\end{figure*}
%
%
\begin{figure*}[tp] \small
\centering
\begin{tabular}{cc}
\epsfig{width=0.95\figurewidth,file=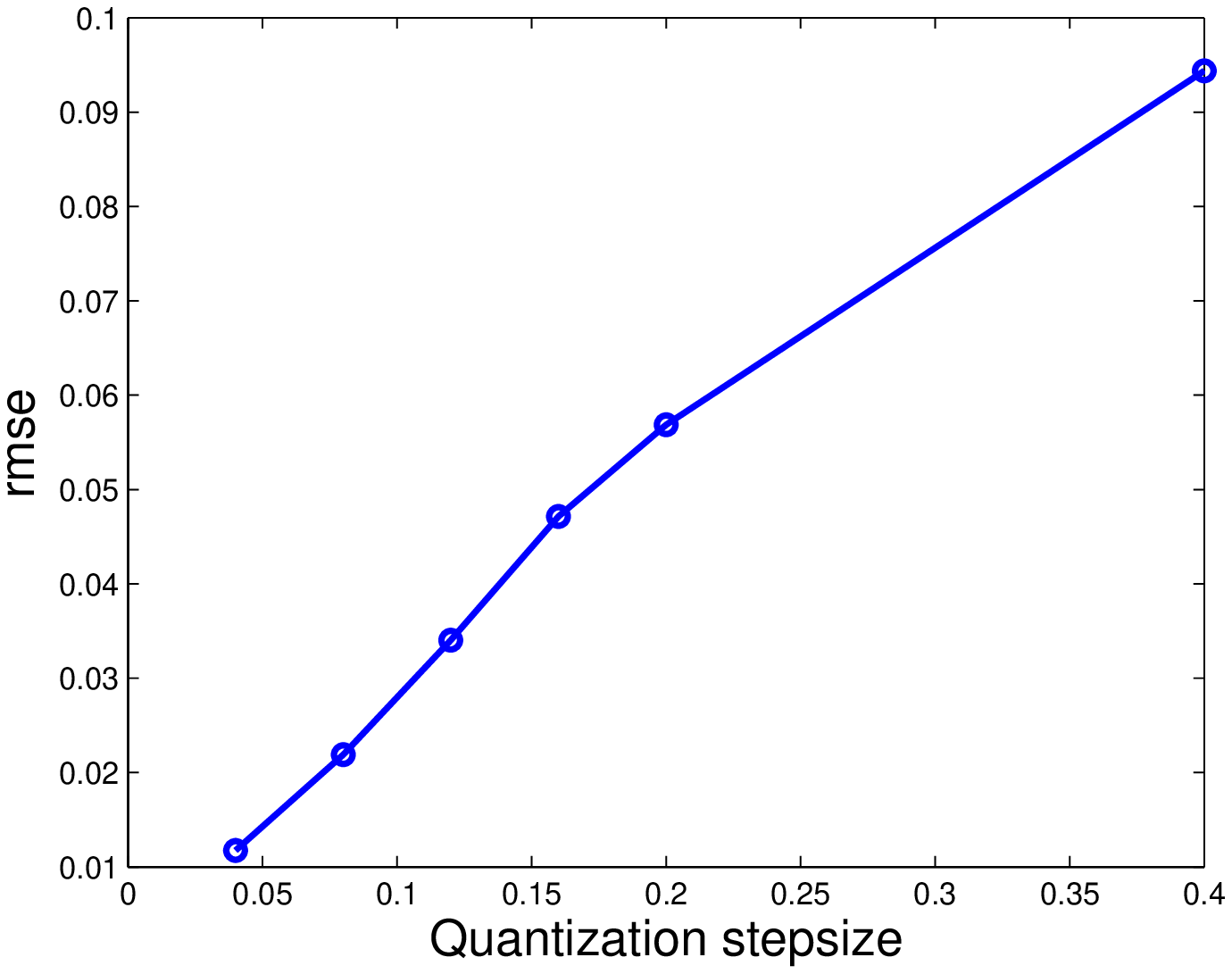}&
   \epsfig{width=0.95\figurewidth,file=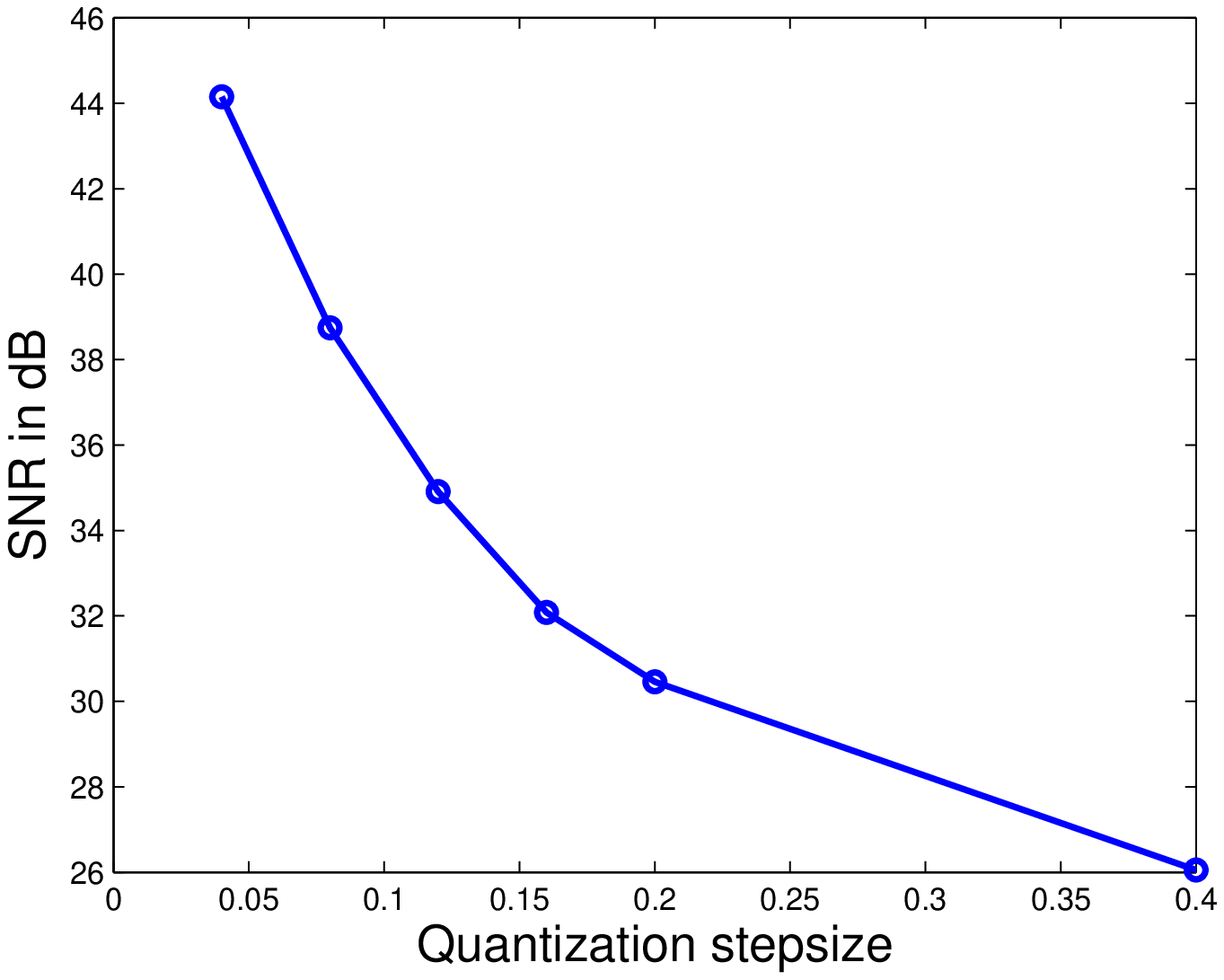}\\
      (a)RMSE of the reconstructed MT signal  vs. q &  (b) SNR of the reconstructed MT signal vs. q   \\[0.5em]
 \epsfig{width=0.95\figurewidth,file=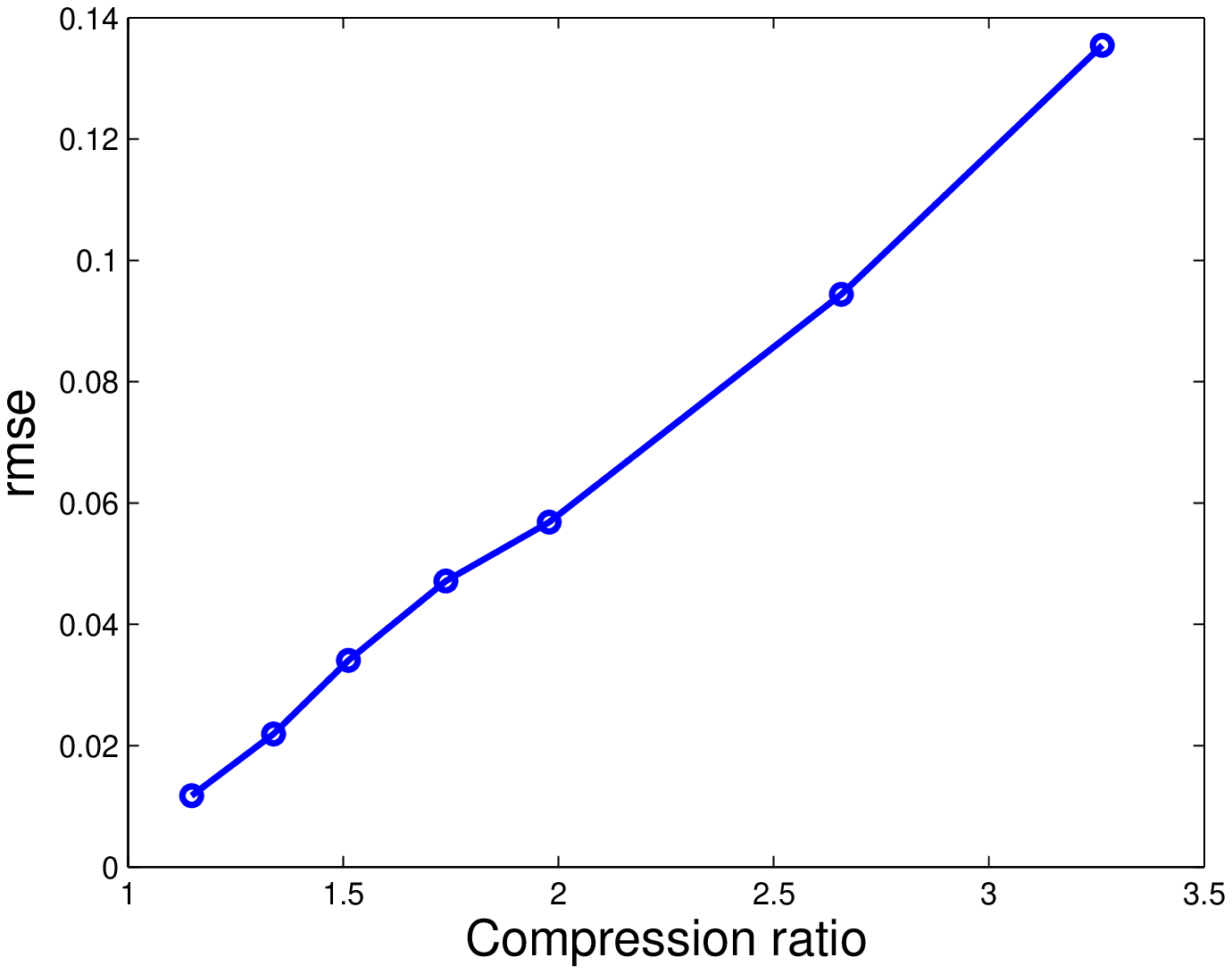}&
 \epsfig{width=0.95\figurewidth,file=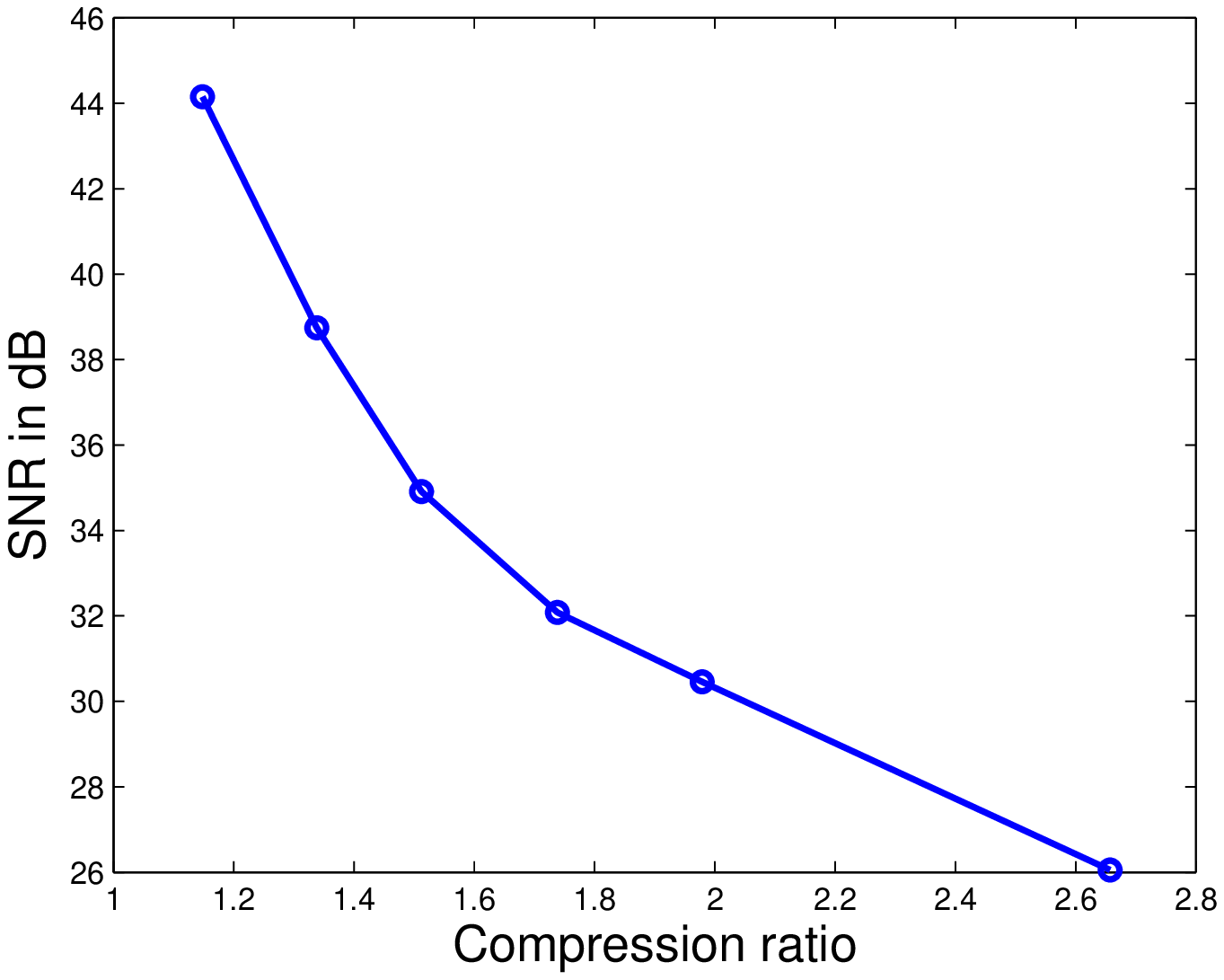}\\
        (c) RMSE of the reconstructed MT signal vs. CR  & (d) SNR of the reconstructed MT signal vs.  CR  \\[0.5em]
 
   \end{tabular}
\vspace*{-0.1in}
\end{figure*}

\begin{figure*}[tp] \small
\centering
\epsfig{width=0.95\figurewidth,file=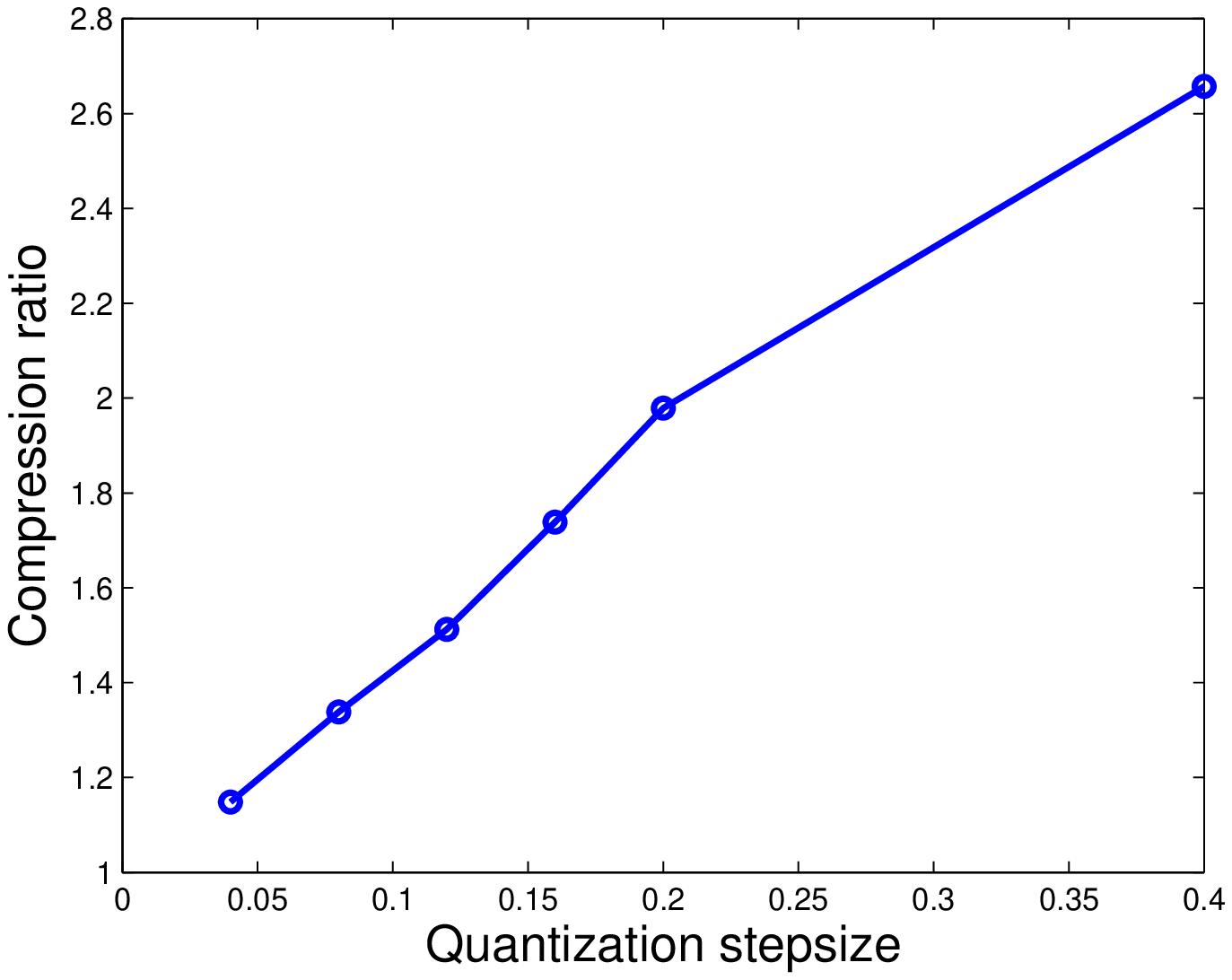}\\
        (e) Variation of CR with respect to quantization step size $q$   \\[0.5em]

\caption{Estimated MT error parameters.}
\label{fig:cr_q}
\end{figure*}

%
\begin{table}[h]
\centering
\begin{tabular}{|lc|c|c|c|c|c|c|}
\hline
\textbf{Parameter data}&\textbf{$f_{sp}$}&\textbf{$f_{ps}$}&\textbf{$f_{gp}$}&\textbf{$f_{pg}$}&\textbf{$f_{sg}$}&\textbf{$f_{gs}$}\\ \hline 
\textbf{Original ABII MT data}& 8.76&0.11&0.11& 8.54& 0.22& 0.22\\ \hline 
\textbf{$q=0.04$}& 11.94&0.11&0&5.47&0.22&0.22 \\ \hline
\textbf{$q=0.08$}& 11.17&0&0&6.57& 0&0.22 \\ \hline
\textbf{$q=0.12$}&  7.01 & 0.11& 0.11&10.51& 0.11& 0.11 \\ \hline
\textbf{$q=0.16$}& 7.01& 0.11&0.22&10.40&  0.11& 0.11 \\ \hline
\textbf{$q=0.2$}& 7.01& 0.11&0.11&10.51& 0.11& 0.11\\ \hline
\textbf{$q=0.4$}& 5.36& 0&0.11&12.26& 0.11&0.11\\ \hline
\end{tabular}
\caption{Transition frequency parameters for the original and reconstructed MT signal for varying values of quantization step size $q$.}
\label{table:param1}
\end{table}
%
\begin{table}[h]
\centering
\begin{tabular}{|lc|c|c|c|}
\hline
\textbf{Parameter data}&\textbf{$v_s$}&\textbf{$v_g$}&\textbf{avg  L}\\ \hline 
\textbf{Original ABII MT data }& 45.17&40.65&5.24  \\ \hline 
\textbf{$q=0.04$}&60.23&40.65&5.19 \\ \hline
\textbf{$q=0.08$}&81.31&36.14&7.61 \\ \hline
\textbf{$q=0.12$}& 60.23& 32.52&8.80 \\ \hline
\textbf{$q=0.16$}& 45.17&40.65&6.53 \\ \hline
\textbf{$q=0.2$}& 60.23&40.65&8.75 \\ \hline
\textbf{$q=0.4$}& 81.31&81.31&15.15 \\ \hline
\end{tabular}
\caption{Velocity  and average length parameters of the original and reconstructed MT signal for varying values of quantization step size $q$}
\label{table:param2}
\end{table}

From a biological standpoint, it is essential to estimate pertaining MT parameters such as transition frequencies, transition peaks, error parameters , and average length of MT to understand the state of the system. Fig \ref{fig:err} illustrates the error estimates of transition peaks and frequencies over different step sizes. Table \ref{table:param1} and Table \ref{table:param2}, gives us the original and estimated parameter values of $3$ state MT signal. From the above tables, we see that, although, the individual transition frequency errors are very small, even a small change in estimated frequencies can cause huge fluctuations in the computed average length $L$ of MT. Hence, it can be regarded that $3$ state MTs are sensitive to the observed parameters. To demonstrate the efficacy and overall performance of our approach we computed Root Mean Square Error (RMSE) and the  Signal to Noise Ratio (SNR) for varying step sizes $q$. The RSME and SNR are defined as follows:\\

\be
RMSE=\sqrt{\frac{1}{N}\sum_{i=1}^{N}(\hat{\bf{x}}_i-\bf{x_i})^2}
\ee

\be
SNR=10log\Bigg(\frac{\sigma^{2}{(\bf{x})}}{{\frac{1}{N}\sum \limits_{i=1}^{N}(\hat{\bf{x}}_i-\bf{x_i})^2}}\Bigg)
\ee
where, $\sigma^{2}(\bf{x})$ denotes the variance of our input MT signal $\bf{x}$ and $\hat{\bf{x}}$ represents recovered MT signal. 

From Fig \ref{fig:cr_q} , it can be inferred that as the quantization step size increases, so does the CR, RMSE and equivalently there is a gradual decay in SNR values. For quantization step sizes $> 0.6$  we see greater fluctuation in the estimated MT parameters and gradual increase in RMSE and decrease in SNR values.  Also, Fig \ref{fig:cr_q}, shows the SNR of reconstructed signal. Due to the dynamic instability of MTs we do not observe higher CRs as in case of other biomedical signals such as ECG. Nevertheless the efficacy of the proposed approach can be further verified by looking at the high SNR and low RMSE of the system. Hence, in this paper we propose an unique instance where compression techniques can be employed enabling us to reduce the number of samples needed for reconstruction of MT signal. However, the trade-off here depends on the applications for which the recovered MT signal would be further used (implying application specific SNR-RMSE compromises)  and  choosing the right quantization step size and in turn achieving desired data compression.

\section{Summary and Conclusion}
\label{sec:conc}
In this work, we proposed a wavelet based non-traditional stack run coding compression framework to reduce the number of samples required for reconstruction of MT signal. Instead of using 4 symbols for encoding data, we have used 5 symbols to incorporate peak information along with the encoding process without causing a compression overhead. The encoded signal is decoded and recovered from wavelet domain. The results prove that we can achieve considerable compression, with a tradeoff in SNR and RMSE values. The dynamic instability of MT hinders higher data compression in comparison to other biomedical signals such as ECG. It was experimentally found that even with low errors in estimated transition frequency parameters, caused greater fluctuation in the average length of MT. Thus, it can be concluded that $3$ state MT signals are sensitive to their parameters. For higher values of step size $q$, we had higher compression, but lead to higher RMSE, lower SNR and lower accuracy in peak detection because of loss of signal details. And conversely, with lower step size $q$, we had higher SNR, lower  RMSE, and higher accuracy in peak detection, because of close approximation of the signal , but lead to lower compression rate. Hence, depending on the further applications that the recovered MT signal might me used, we might have to compromise between desirable compression rate and preferred SNR and RMSE values. For future work, further investigation can be conducted by using vector quantization instead of uniform quantization and employing adaptive arithmetic coders  for compression of MT signals.

\end{document}